\documentclass[twoside,11pt]{entics} 
\usepackage{enticsmacro}
\usepackage{graphicx}
\usepackage[all]{xy}
\usepackage{amsmath,amssymb}
\newcommand\qedhere{\vspace{-0.9em}}
\usepackage{mdwlist}
\usepackage{xspace}
\usepackage{fancyvrb}
\usepackage{doi}
\newcommand\broadcast{\tf{bcst}}
\newcommand\echo{\tf{echo}}
\newcommand\ready{\tf{ready}}
\newcommand\deliver{\tf{dlvr}}

\newcommand\fvote{\f{vote}}
\newcommand\fdecide{\f{observe}}

\newcommand\figunderline[1]{\text{\rlap{\underline{\emph{#1}}}}}
\newcommand\figskip{\\}
\newcommand\rulefont[1]{\ensuremath{{\mathsf{(#1)}}}}

\newcommand\compressthislight[1]{{\hspace{.8pt}\raisebox{.5pt}{\scalebox{.85}{$#1$}}\hspace{.2pt}}}
\newcommand\compressthis[1]{\pmb{\compressthislight{#1}}}

\newcommand\tneg{{\pmb\neg}}

\newcommand\tbot{{\pmb\bot}}
\newcommand\botval{\text{\textonehalf}} 
\newcommand\teq{{\hspace{0.5pt}\pmb{\scalebox{0.7}[1]{=}}}}

\newcommand\tand{\mathrel{\pmb\wedge}}
\newcommand\tor{\mathbin{\pmb\vee}}
\newcommand\toplus{\mathbin{\pmb\oplus}}

\newcommand\predsymb{\ns{PredSymb}}

\newcommand\limp{\Longrightarrow}

\newcommand\tnotor{\mathbin{\rightarrow}}

\newcommand\timpc{\mathbin{\rightharpoonup}}

\newcommand\tall{{\compressthis{\forall}}}
\newcommand\texi{{\compressthis{\exists}}}
\newcommand\texiunique{{\compressthis{\existunique}}}
\newcommand\texiaffine{{\compressthis{\existaffine}}}
\newcommand\tvequal{=}

\newcommand\f[1]{\mathit{#1}}
\newcommand\tf[1]{\mathsf{#1}}

\newcommand\ns[1]{\mathsf{#1}}
\newcommand\plus{{+}}
\newcommand\minus{{\text{-}}}

\newcommand\lmodel{[\hspace{-0.2em}[}
\newcommand\rmodel{]\hspace{-0.2em}]}
\newcommand\model[1]{{\lmodel #1 \rmodel}}
\newcommand\modellabel[2]{{\lmodel #1 \rmodel}_{#2}}
\newcommand\modelm[1]{\modellabel{#1}{\mu}}
\newcommand\cent{\vdash}
\newcommand\ment{\vDash}
\newcommand\mentM{\ment_{\hspace{-2pt}\mu}}

\newcommand\nment{\mathrel{\not\vDash}}

\newcommand\Forall[1]{\forall #1.}
\newcommand\Exists[1]{\exists #1.}
\newcommand\ssm{{{:}\text{=}}}

\newcommand\existunique{\exists_{1}}
\newcommand\existaffine{\exists_{01}}

\usepackage{tikz}

\newcommand\points{\ns{Pnt}}
\newcommand\Val{\ns{Val}}

\newcommand\opens{\ns{Opn}}
\newcommand\opensne{\opens_{\hspace{-1pt}{\neq}\varnothing}}

\newcommand\weakmodusponens{weak modus ponens (Prop.~\ref{prop.mp.for.tnotor}(\ref{item.mp.for.tnotor}))\xspace}
\newcommand\strongmodusponens{strong modus ponens (Prop.~\ref{prop.mp.for.tnotor}(\ref{item.mp.for.timpc}))\xspace}
\newcommand\weakmodusponensnoref{weak modus ponens\xspace}

\newcommand\THREE{{\mathbf 3}}

\newcommand\tvT{{\mathbf t}}
\newcommand\tvF{{\mathbf f}}
\newcommand\tvB{{\mathbf b}}
\newcommand\modality[1]{#1\hspace{1.5pt}}
\newcommand\modT{\modality{\tf T}}
\newcommand\modTB{\modality{\tf T\hspace{-2.5pt}\tf B}}

\newcommand\modB{\modality{\tf B}}
\newcommand\modTF{\modality{\tf T\hspace{-2.5pt}\tf F}}
\newcommand\modF{\modality{\tf F}}
\newcommand\everyone{\modality{\Box}}
\newcommand\someone{\modality{\Diamond}}
\newcommand\someoneAll{\someone}
\newcommand\somewhere{\someone}
\newcommand\everyoneAll{\everyone}
\newcommand\everywhere{\everyone}

\usepackage{amssymb}
\makeatletter
\newcommand{\dotop}[1]{{\mathpalette\dotop@{#1}}}
\newcommand{\dotop@}[2]{%
  \vphantom{#2}%
  \ooalign{$\m@th#1\mathop#2$\cr\hidewidth$\m@th#1\cdot$\hidewidth\cr}%
}
\makeatother
\newcommand\jamiediamond{{\raisebox{-0.45pt}{\scalebox{1.35}[1]{$\lozenge$}}}}

\newcommand{\lddiamond}{\dotop{\jamiediamond}}
\newcommand{\ldsquare}{\raisebox{.7pt}{$\dotop{\square}$}}

\newcommand\Quorum{\ldsquare}
\newcommand\quorum{\Quorum}
\newcommand\QuorumBox{\Quorum}
\newcommand\Coquorum{\lddiamond}
\newcommand\Contraquorum{\Coquorum}
\newcommand\contraquorum{\Contraquorum}
\newcommand\CoquorumDiamond{\Coquorum}

\newcommand\correct[1]{\modTF\hspace{-2pt}[#1]}
\newcommand\incorrect[1]{\modB\hspace{-2pt}[#1]}
\newcommand\powerset{\f{powerset}}

\newcommand\deffont[1]{{\bfseries #1}}
\usepackage{slantsc}
\newcommand\theory[1]{\ensuremath{\text{\itshape\scshape #1}}\xspace}

\newcommand\ThyBB{\theory{ThyBB}}
\newcommand\ThyCA{\theory{ThyCA}}
\newcommand\ThyVote{\theory{ThyVote}}

\newenvironment{thrm}{\begin{theorem}}{\end{theorem}}
\newenvironment{lemm}{\begin{lemma}}{\end{lemma}}
\newenvironment{rmrk}{\begin{remark}}{\end{remark}}
\newenvironment{xmpl}{\begin{example}}{\end{example}}
\newenvironment{defn}{\begin{defnn}\upshape}{\end{defnn}}
\newenvironment{corr}{\begin{corollary}\upshape}{\end{corollary}}
\newenvironment{nttn}{\begin{nttnn}\upshape}{\end{nttnn}}

\volume{NN}			%
\def\lastname{Murdoch J. Gabbay} %
\def\runtitle{Declarative distributed algorithms} %
\begin{document}
\begin{frontmatter}
\vspace{-6em} %
\title{Declarative distributed algorithms as axiomatic theories in three-valued modal logic over semitopologies}
  \author{Murdoch J. Gabbay} %
  \address{{\normalfont\href{http://www.gabbay.org.uk}{gabbay.org.uk}}}
\begin{abstract} 
We illustrate how to formally specify distributed algorithms as declarative axiomatic theories in a modal logic, using as illustrative examples a simple voting protocol, a simple broadcast protocol (Bracha Broadcast), and a simple agreement protocol (Crusader Agreement).  
The methods scale well and have been used to find errors in a proposed industrial protocol.

The key novelty is to use modal logic to capture a declarative, high-level representation of essential system properties -- the logical essence of the algorithm -- while abstracting away from explicit state transitions of an abstract machine that implements it.
It is like the difference between specifying code in a functional or logic programming language, versus specifying code in an imperative one.
Thus we present axiomatisations of \emph{Declarative Bracha Broadcast} and \emph{Declarative Crusader Agreement}. 

A logical axiomatisation in the style we propose provides a precise, compact, human-readable specification that abstractly captures essential system properties, while eliding low-level implementation details; it is more precise than a natural language description, yet more abstract than source code or a logical specification thereof.
This creates new opportunities for reasoning about correctness, resilience, and failure, and could serve as a foundation for human- and machine verification efforts, design improvements, and even alternative protocol implementations. 

The proofs in this paper have been formalised in Lean~4.%

\end{abstract}
\tableofcontents
\end{frontmatter}
\section{Introduction}
\label{sect.intro}

\subsection{Introducing the declarative, axiomatic approach}

We offer a new approach to analysing distributed algorithms, by exhibiting them as axiomatic theories in a purpose-built three-valued modal logic (one with a `middle' truth-value $\tvB$ between $\tvT$ and $\tvF$) over a semitopology (like topology but without the condition that the intersection of two open sets is necessarily open)~\cite{gabbay:semdca,gabbay:semtad}.
This captures a declarative essence of the essential system properties.

Nonempty open sets of a semitopology represent quorums.
If a predicate returns $\tvT$ (true) or $\tvF$ (false) this reflects correct behaviour, and if it returns $\tvB$ (both / byzantine) this reflects faulty behaviour.

The effect is similar to the difference between an explicit for-loop to apply a function \verb+f+ to each element of a list \verb+l+, and just writing \verb+map f l+: we abstract declaratively to the essential logic of the algorithm.

This paper is designed to introduce the basic ideas of this new declarative / axiomatic approach.
These techniques can be applied to larger examples, including modern industrial-scale ones, but here we consider simpler algorithms so that we get a clearer view of the basic idea.
We consider Bracha Broadcast and Crusader Agreement: these are short but non-trivial examples of the genre; they illustrate our techniques; and this provides an accessible introduction to and exposition on this new way to apply logic to study distributed algorithms. 

In this paper, our examples (voting, broadcast, agreement) will be toy and/or textbook algorithms.
These toy examples will not, and are not intended to, showcase the full power of our proposed approach relative to standard distributed computing methods.
This is a feature, not a bug: we show how to attain simplicity via abstraction in logic, using familiar examples as case studies.
Attaining this simplicity is mathematically non-trivial, and we shall see that there is plenty of subtlety to discuss along the way (and much mathematical beauty, too).

The proofs in this paper have been formalised in Lean~4 and sources are freely available for download~\cite{rovira:hbbl4}.

As for scaling up, the techniques illustrated in this paper scale well. 
In fact, more complex algorithms create \emph{more} scope to benefit from the abstractions in this paper by eliding implementational detail:
\begin{enumerate*}
\item
In other work we apply these ideas to Paxos~\cite{gabbay:decasd}.\footnote{The logic and models needed to do this are different, reflecting that Paxos is a different algorithm, but the basic `declarative' principle is similar.}
\item
Since this paper was conceived, we scaled up declarative techniques to analyse \emph{and find errors in} a proposed industrial protocol called \emph{Heterogeneous Paxos}~\cite{sheff:hetp,sheff_et_al:LIPIcs.OPODIS.2020.5}.
The complexity of Heterogeneous Paxos %
had made it prohibitively hard to debug using traditional, imperative techniques. 
We completed a declarative analysis, found the error, and we used declarative techniques to design a replacement protocol~\cite{gabbay:hettrb}, whose correctness was formally verified in Lean~4~\cite{hart:hetbl4} by Hart.\footnote{Hart also worked on the original pseudocode/English specification.  He reported that he found the axiomatic/declarative specification and proofs much easier to work with.}
\end{enumerate*}

\emph{A note on terminology:}\ we will use `distributed algorithm' and `distributed protocol' synonymously.
The difference is in emphasis: if we write `algorithm' we imply an omniscient view of the system from the outside, whereas if we write `protocol' we imply a view of the system from a participant communicating with other participants according to the rules of the algorithm.
However, this difference in emphasis just reflects a switch in point of view.
The reader can safely identify both words with the same meaning throughout.

\subsection{Motivation: distributed algorithms}

Distributed algorithms are algorithms that run across multiple participants, such that the overall behaviour of the system emerges from their cooperation.
Contrast with a \emph{parallel} algorithm, which also runs across multiple participants but which would still make sense running (albeit more slowly) on a single participant.\footnote{A GPU rendering algorithm is a good example of a typical parallel algorithm.  It could run on a single-threaded CPU, and that would make perfectly good sense; it would just be too slow.
In contrast, a blockchain is a good example of an algorithm that is fundamentally distributed and not parallel: it could still run on just one single participant, but that would completely defeat the point of the blockchain, which is precisely to \emph{be} distributed across many participants.}
Distribution can provide benefits in resilience and scalability which could not be obtained any other way.

A fundamental challenge of distribution is that we cannot assume that all participants in the algorithm are correct: we need to be resilient against failure, and against hostile (also called \emph{byzantine}) behaviour.
Distributed algorithms are notoriously difficult to design, specify, reason about, and validate, because of the inherent difficulty of coordinating multiple participants who may be faulty, hostile, drifting in and out of network connectivity~\cite{sleepy,santoro:timnh}, and have different local clocks, states, agendas, etc.

This gives the field of distributed algorithms a very particular flavour.
Everything revolves around this central question: how does our algorithm guarantee prompt and reliable results for correct participants, when parts of the system may be faulty? 

Blockchain consensus algorithms are a good source of examples, both for distributed algorithms and for how things can go wrong.\footnote{This work was motivated by analysis of blockchain systems.} 
Solana (\url{https://solana.com}) has experienced instances where its consensus mechanism contributed to temporary network halts~\cite{yakovenko2018solana,shoup2022poh}; Ripple's (\url{https://ripple.com}) XRP consensus protocol has been analysed for conditions in which safety and liveness may not be guaranteed~\cite{DBLP:conf/opodis/Amores-SesarCM20};
and Avalanche's (\url{https://www.avax.network}) consensus protocol has faced scrutiny regarding its assumptions and performance in adversarial scenarios~\cite{DBLP:conf/opodis/Amores-SesarCT22}.
Even mature blockchains like Ethereum are prone to this: e.g. \emph{balancing attacks}~\cite{DBLP:conf/sp/NeuTT21} and \emph{ex-ante reorganisation} (\emph{reorg}) attacks~\cite{DBLP:conf/fc/Schwarz-Schilling22} exploit vulnerabilities in Ethereum's current protocol.
Even as this paper was prepared, XRP consensus experienced a slowdown and stoppage, which was addressed in a GitHub commit.\footnote{\href{https://cryptodamus.io/en/articles/news/xrpl-meltdown-64-minute-outage-xrp-ledger-s-february-4th-2025-halt-decoded}{cryptodamus.io/en/articles/news/xrpl-meltdown-64-minute-outage-xrp-ledger-s-february-4th-2025-halt-decoded} (\href{https://web.archive.org/web/20251224082741/https://cryptodamus.io/en/articles/news/xrpl-meltdown-64-minute-outage-xrp-ledger-s-february-4th-2025-halt-decoded}{permalink}) and \href{https://github.com/XRPLF/rippled/pull/5277}{github.com/XRPLF/rippled/pull/5277} (\href{https://wayback-api.archive.org/web/20251224083444/https://github.com/XRPLF/rippled/pull/5277}{permalink}).}

Thus, subtle flaws and/or ambiguities in protocol design have significant consequences, emphasising the need for rigorous analysis and well-defined specifications in this complex and safety-critical field.

Meanwhile in the real world, such analysis and specification may be abbreviated or skipped altogether: that is, in industrial practice a protocol may go from an initial idea written in English or pseudocode to an implementation, without a rigorous analysis.
This is because such analysis would be just too expensive.
Part of the motivation for this work is, via logic, to provide a precise, compact, nontrivial, yet (relatively) lightweight methodology for specifying and reasoning about distributed protocols. 
Returning to the example above of an explicit for-loop versus a map function, the idea is to attain compact proofs by being as declarative and static as possible, while preserving the essence of the algorithm. 

\subsection{Map of the paper}

\begin{itemize*}
\item
We give a simple voting algorithm in Section~\ref{sect.illustrative.example}.
This gives us an opportunity to introduce semitopologies and three-valued logic.
\item
We define the syntax of a modal logic in Section~\ref{sect.logic}: this is the formal language which we will use for our axioms and correctness properties.
\item
We study Bracha Broadcast in Section~\ref{sect.bb}.
\item
We study Crusader Agreement in Section~\ref{sect.crusader.agreement}.
\item
Finally, we conclude with Section~\ref{sect.conclusions}.
\end{itemize*}

\section{Motivation and an illustrative example: voting}
\label{sect.illustrative.example}

We will warm up by studying a voting algorithm. 
This is simpler than Bracha Broadcast but is still enough to invoke much of our machinery, including the three-valued logic and semitopologies.\footnote{We defer a treatment of quantifiers and formal logical syntax to the case of Bracha Broadcast.}
The axiomatisation is in Figure~\ref{fig.simple}.
In Example~\ref{xmpl.vote} we describe the protocol informally:

\subsection{Protocol description}
 
\begin{xmpl}
\label{xmpl.vote}
Fix $f\geq 1$ and consider a set $\points$ of $3f\plus 1$ \deffont{participants}. 
\begin{itemize*}
\item
Call a subset $Q\subseteq\points$ a \deffont{quorum} when $\#Q\geq 2f\plus 1$; 
\item
call $Q$ a \deffont{coquorum} when $\#Q\leq f$; and 
\item
call $Q$ a \deffont{contraquorum} or \deffont{blocking set} when $\#Q\geq f\plus 1$.
\end{itemize*}
A protocol to hold a ballot to choose between \emph{true} $\tvT$ and \emph{false} $\tvF$ proceeds as follows:
\begin{enumerate*}
\item
Each participant $p\in\points$ \deffont{broadcasts} to all participants a \deffont{vote} message carrying a payload $\tvF$ or $\tvT$. 
\item
Then:
\begin{enumerate*}
\item
If $p$ receives vote-$\tvT$ messages from a quorum of participants, then $p$ deems that the ballot has succeeded and $p$ \deffont{observes} $\tvT$.
\item
If $p$ receives vote-$\tvF$ messages from a quorum of participants, then $p$ deems that the ballot has succeeded and $p$ \deffont{observes} $\tvF$.
\item
If $p$ receives no quorum of messages voting for $\tvT$ and no quorum of messages voting for $\tvF$, then $p$ deems that the ballot has failed, and $p$ observes no value.
\end{enumerate*}
\end{enumerate*} 
\end{xmpl}

\begin{rmrk}
\label{rmrk.failure.assumptions.voting}
We need to choose our failure assumptions: 
\begin{itemize*}
\item
Do messages always arrive?
\item
Can any participants crash? 
\item
Do all participants follow the algorithm, or could some be dishonest/faulty/byzantine/hostile (= not follow the algorithm)?
\end{itemize*}
There is no right or wrong here: we just need to be clear about our assumptions.\footnote{%
For convenient reference, I include here a brief survey of some of the terminology used in the distributed systems and blockchain literature.

Call a participant in a distributed algorithm \emph{correct} or \emph{honest} when they follow the protocol of the algorithm.
These are synonymous, but the former suggests that we view the participant as a machine, whereas the latter suggests a participant with agency to choose whether to misbehave (e.g. in a blockchain system).
Call a participant that does not follow the protocol \emph{faulty} (for machines) or \emph{byzantine} or \emph{adversarial} or \emph{hostile} (for participants with agency). 

Types of fault include \emph{crash fault} (stops sending messages), \emph{byzantine fault} (may deviate arbitrarily from the protocol), or \emph{omission fault} (observes the protocol when live, but might choose to fake a crash).  

In real life, it is also possible to play timing games (e.g. acting at the last possible moment).
These are all `honest' in a technical sense, though not necessarily in a social or legal sense; e.g. %
\emph{front running} in trading is illegal.  That is out of scope for this paper. 
}

For this exercise our failure assumptions will be as follows: 
\begin{itemize*}
\item
message-passing is reliable (messages always arrive; they do not get lost or delayed); 
\item
participants never crash; 
\item
participants always send messages when they should; however, some participants may be dishonest and send vote-$\tvT$ messages to some participants and vote-$\tvF$ messages to others, contrary to the protocol which specifies they should broadcast the \emph{same} vote value to \emph{all} participants.
\end{itemize*}
This might lead to a situation in which honest participants disagree about which vote won the ballot.
We want to make sure that this is impossible; or to be more precise, we want to know what conditions are sufficient to ensure Agreement (this property is also called Consistency): 
\begin{quote}
\emph{all honest participants observe the same value, if they observe any value at all}.
\end{quote}
We can prove Agreement/Consistency under an assumption that there is a quorum $Q$ of honest participants, 
i.e. that there are at most $f$ dishonest participants: 
Observe by a counting argument that %
any two quorums $Q_1,Q_2\subseteq\points$ of votes (each of which contains at least $2f\plus 1$ participants) must intersect $Q$ on at least one participant $q\in Q\cap Q_1\cap Q_2$.
Since $q\in Q$, it is honest. 
This means it is impossible for a participant $p$ to observe $\#Q_1\geq 2f\plus 1$ vote messages for $\tvT$, while at the same time \emph{another} participant $p'$ observes $\#Q_2\geq 2f\plus 1$ vote messages for $\tvF$, because one participant is honest and votes either for $\tvT$ or $\tvF$.
\end{rmrk}

We now show how to declaratively model the protocol and the proof of Agreement.
First, we will need two ingredients: a notion of \emph{truth}, which we define in Subsection~\ref{subsect.3}; and a notion of \emph{quorum}, which we define in Subsection~\ref{subsect.semitopology}.

\subsection{Three-valued logic}
\label{subsect.3}

\newenvironment{ttarray}{\begin{array}{c@{\ \ }c c c}}{\end{array}}
\newenvironment{tarray}{\begin{array}{c@{\hspace{-0pt}}c}}{\end{array}}

\begin{figure} %
$$
\hspace{-1em}
\begin{array}{c}
\begin{array}{c@{\quad}c@{\quad}c@{\quad}c@{\quad}c@{\quad}c@{\quad}c@{\quad}c@{\quad}c@{\quad}c@{\quad}c}
\begin{ttarray}
p\tand q & \tvT & \tvB & \tvF
\\
\tvT &  \tvT & \tvB & \tvF
\\
\tvB &  \tvB & \tvB & \tvF 
\\
\tvF &  \tvF & \tvF & \tvF 
\end{ttarray}
&
\begin{ttarray}
p\tor q & \tvT & \tvB & \tvF
\\
\tvT &  \tvT & \tvT & \tvT
\\
\tvB &  \tvT & \tvB & \tvB
\\
\tvF &  \tvT & \tvB & \tvF
\end{ttarray}
&
\begin{ttarray}
p\toplus q & \tvT & \tvB & \tvF
\\
\tvT &  \tvF & \tvB & \tvT
\\
\tvB &  \tvB & \tvB & \tvB
\\
\tvF &  \tvT & \tvB & \tvF
\end{ttarray}
&
\begin{ttarray}
p\tnotor q & \tvT & \tvB & \tvF
\\
\tvT &  \tvT & \tvB & \tvF
\\
\tvB &  \tvT & \tvB & \tvB
\\
\tvF &  \tvT & \tvT & \tvT
\end{ttarray}
&
\begin{ttarray}
p \timpc q & \tvT & \tvB & \tvF
\\
\tvT       & \tvT & \tvF & \tvF
\\
\tvB       & \tvT & \tvB & \tvB
\\
\tvF       & \tvT & \tvT & \tvT
\end{ttarray}
\end{array}
\\[8ex]
\begin{array}{c@{\quad}c@{\quad}c@{\quad}c@{\quad}c@{\quad}c@{\quad}c@{\quad}c@{\quad}c@{\quad}c@{\quad}c}
\begin{tarray}
\tneg p
\\
\tvT & \tvF
\\
\tvB & \tvB
\\
\tvF & \tvT
\end{tarray}
&
\begin{tarray}
\modT p 
\\
\tvT & \tvT
\\
\tvB & \tvF
\\
\tvF & \tvF
\end{tarray}
&
\begin{array}{cc}
\modF p 
\\
\tvT & \tvF
\\
\tvB & \tvF
\\
\tvF & \tvT
\end{array}
&
\begin{tarray}
\modTB p  
\\
\tvT & \tvT
\\
\tvB & \tvT
\\
\tvF & \tvF
\end{tarray}
&
\begin{tarray}
\modTF p  
\\
\tvT & \tvT
\\
\tvB & \tvF
\\
\tvF & \tvT
\end{tarray}
&
\begin{tarray}
\modB p
\\
\tvT & \tvF
\\
\tvB & \tvT
\\
\tvF & \tvF
\end{tarray}
\end{array}
\end{array}
$$
\emph{Vertical axis above indicates values of $p$; horizontal axis (if nontrivial) indicates values of $q$.}
\caption{Truth-tables for three-valued connectives on $\THREE=\{\tvT,\tvB,\tvF\}$ (Definition~\ref{defn.THREE}(\ref{item.THREE.connectives}))}
\label{fig.3}
\end{figure}

Our notion of truth will be three-valued:
\begin{defn}
\label{defn.THREE}
\leavevmode
\begin{enumerate*}
\item
Let $\THREE=(\tvF<\tvB<\tvT)$ be a lattice of \deffont{truth-values} with elements 
\begin{itemize*}
\item
$\tvF$ (false), which is less than 
\item
$\tvB$ (Both/Byzantine/Broken/amBivalent), which is less than 
\item
$\tvT$ (true).
\end{itemize*}
\item\label{item.THREE.connectives}
In Figure~\ref{fig.3} we define:
\begin{itemize*}
\item
a unary connective $\tneg$ (\deffont{negation}), and 
\item
binary connectives $\tand$ (\deffont{conjunction}), $\tor$ (\deffont{disjunction}), $\toplus$ (\deffont{exclusive-or}), $\tnotor$ (\deffont{weak implication}), and $\timpc$ (\deffont{strong implication}), %
\item
modalities $\modT$, $\modB$, $\modF$, $\modTB$, and $\modTF$.\footnote{We call $\tneg$ a `connective' and (for example) $\modT$ a `modality'.  This is just a matter of tradition: structurally, they are both just unary operators.}
\end{itemize*} 
\item\label{item.THREE.pointwise}
Suppose $X$ is a set.
We may lift connectives from $\THREE$ to the function-space $X\to\THREE$ in the natural \deffont{pointwise} definition.
Thus: if $f,f':X\to\THREE$ then we may use notation as follows: 
$$
\begin{array}{r@{\ }l}
\tneg f =& \lambda x.\tneg f(x)
\\
f\tand f' =& \lambda x.f(x)\tand f'(x)
\\
f\tor f' =& \lambda x.f(x)\tor f'(x)
\end{array}
$$
\item\label{item.valid}
If $\f{tv}\in\THREE$ then write $\ment\f{tv}$ for the judgement `$\f{tv}\in\{\tvT,\tvB\}$' and read this as $\f{tv}$ is \deffont{valid}.\footnote{There is a standard distinction in logic between a truth-value being \emph{true} (which means: it is equal to $\tvT$) and a logical judgement being \emph{valid} (which may mean `a true assertion about a model' or `it is derivable in some derivation system' or `it evaluates to true in a given context').
In two-valued logic this distinction is less prominent, because a predicate may be judged to be valid when it evaluates to $\tvT$ (possibly in the context of some valuation and/or interpretation).  But here, we have three truth-values, and a predicate will be judged to be valid when it evaluates to $\tvT$ \emph{or} $\tvB$.}
Write $\nment\f{tv}$ for `$\f{tv}=\tvF$' (this happens precisely when $\neg(\ment\f{tv})$) and read this as $\f{tv}$ is \deffont{invalid} or is \deffont{false}.
\end{enumerate*}
\end{defn}

\begin{rmrk}
\label{rmrk.3.core}
Note that the connectives in Definition~\ref{defn.THREE}(\ref{item.THREE.connectives}) and Figure~\ref{fig.3} are not minimal.
We can express them all from $\tneg$, $\tand$, and $\modT$ as per Figure~\ref{fig.3.core}.
\end{rmrk}

\begin{figure}
$$
\begin{array}{r@{\ }l@{\qquad}r@{\ }l} %
\f{tv}\tor\f{tv}' = &\tneg(\tneg\f{tv}\tand\tneg\f{tv}')
&
\f{tv}\toplus\f{tv}' = &(\f{tv}\tand\tneg\f{tv}')\tor(\tneg\f{tv}\tand\f{tv}')
\\
\f{tv}\tnotor\f{tv}' = &(\tneg\f{tv})\tor\f{tv}'
&
\f{tv}\timpc\f{tv}' = &\f{tv}\tnotor\modT\f{tv}'
\\
\modF\f{tv} = &\modT\tneg\f{tv}
&
\modTB\f{tv} = &\tneg\modF\f{tv}
\\
\modTF\f{tv} = &\modT\f{tv}\tor\modF\f{tv}
&
\modB\f{tv} = &\tneg\modTF\f{tv}
\end{array}
$$
\caption{Connectives derived from $\tneg$, $\tand$, and $\modT$ (Remark~\ref{rmrk.3.core})}
\label{fig.3.core}
\end{figure}

The implicational structure of $\THREE$ is rich (it supports sixteen implications~\cite[note~5, page~22]{arieli:idepl}).
For us, the weak implication $\tnotor$ and strong implication $\timpc$ will suffice, and we have: 
\begin{prop}
\label{prop.mp.for.tnotor}
\label{prop.tand.tor}
For truth-values $\f{tv},\f{tv}'\in\THREE$, we have:
\begin{enumerate*}
\item\label{item.tand.tor.tv}
$\ment\tvT$ and $\ment\tvB$ and $\nment\tvF$.
\item\label{item.tand.tor.tand}
$\ment\f{tv}\tor\f{tv}'$ if and only if ${\ment\f{tv}}\lor{\ment\f{tv}'}$, and 
\\
$\ment\f{tv}\tand\f{tv}'$ if and only if ${\ment\f{tv}}\land{\ment\f{tv}'}$.
\item
$(\f{tv}\tnotor\f{tv}')\tvequal(\tneg\f{tv}\tor\f{tv}')$ and
$(\f{tv}\timpc\f{tv}')\tvequal(\f{tv}\tnotor\modT\f{tv}')$.
\item\label{item.mp.for.tnotor}
$\ment\f{tv}\tnotor\f{tv}'$ if and only if 
$\f{tv}\tvequal\tvT$ implies $\f{tv}'\in\{\tvT,\tvB\}$.
We call this \deffont{weak modus ponens}. 
\item\label{item.mp.for.timpc}
$\ment\f{tv}\timpc\f{tv}'$ if and only if 
$\f{tv}\tvequal\tvT$ implies $\f{tv}'\tvequal\tvT$. 
We call this \deffont{strong modus ponens}. 
\item\label{item.para.em}
$\ment\f{tv}\tor\tneg\f{tv}$.
In jargon, validity satisfies \deffont{excluded middle} ($p$-or-not-$p$ is valid).
\item\label{item.para.para}
$\ment\f{tv}\tand\tneg\f{tv}$ if and only if $\f{tv}=\tvB$ if and only if $\ment\modB\f{tv}$.\footnote{This means that our notion of validity is \deffont{paraconsistent}, so $\phi$ and not-$\phi$ may both be valid (\cite{arieli:idepl}, or see the nice survey in~\cite{priest:parl}). Note that in our setting $\phi$ and not-$\phi$ cannot be simultaneously \emph{true}, since $\tneg\tvT=\tvF$.} 
\item\label{item.para.ment}
$\ment\f{tv}$ if and only if $\modTF \f{tv}\tvequal\modT\f{tv}$. 
\item\label{item.para.de.Morgan}
$\f{tv}\leq \f{tv}'$ if and only if $\tneg\f{tv}'\leq\tneg\f{tv}$.
\end{enumerate*}
\end{prop}
\begin{proof}
By inspection of the truth-tables in Figure~\ref{fig.3}.
\end{proof}

\subsection{Semitopologies}
\label{subsect.semitopology}

\begin{figure}
$$
\begin{array}{r@{\quad}l@{\qquad}r@{\quad}l}
\rulefont{Observe?}&
\fdecide(p) \tnotor \Quorum\fvote
&
\rulefont{Observe\tneg?}&
\tneg\fdecide(p) \tnotor \Quorum\tneg\fvote
\\
\rulefont{Observe!}&
\Quorum\fvote\timpc\fdecide(p)
&
\rulefont{Observe\tneg!}&
\Quorum\tneg\fvote\timpc\tneg\fdecide(p)
\\
\rulefont{Correct}&
\Quorum(\modTF\circ\fvote)
&
\rulefont{3twined}&
(\Quorum f\tand\Quorum f')\leq\Contraquorum(f\tand f')
\end{array}
$$

Above: $\circ$ denotes function composition; $f$ and $f'$ range over arbitrary functions in $\points\to\THREE$ and $p$ ranges over elements of $\points$.
\ $\leq$ refers to the lattice ordering on $\THREE$, namely $\tvF<\tvB<\tvT$.
\ $\Quorum$ and $\Contraquorum$ have type $(\points\to\THREE)\to\THREE$ here, which is why $\f{observe}$ is applied to $p$ above and $\f{vote}$ is not.

$\f{vote}(p)=\tvT$ means `$p$ voted for $\tvT$'; $\f{vote}(p)=\tvF$ means `$p$ voted for $\tvF$'; and $\f{vote}(p)=\tvB$ means `$p$ sent conflicted messages about its vote'.
Similarly for $\f{observe}$.
\caption{Specification of a simple voting protocol (Definition~\ref{defn.voting})}
\label{fig.simple}
\end{figure}

We define our universe of participants.
This has a topological flavour:
\begin{defn}
\label{defn.semitopology}
A \deffont{semitopology}~\cite{gabbay:semdca,gabbay:semtad} is a pair $(\points,\opens)$ of 
\begin{enumerate*}
\item\label{item.semitopology.points}
a nonempty\footnote{Cf. Remark~\ref{rmrk.nonempty.P}.}
set of \deffont{points} $\points$ -- which we may call \deffont{participants}, because we may think of them as participants running a distributed algorithm -- and 
\item\label{item.semitopology.opens}
\deffont{open sets} $\opens\subseteq\powerset(\points)$, such that $\points\in\opens$ and $\opens$ is closed under arbitrary (possibly empty) unions: if $\mathcal O\subseteq\opens$ then $\bigcup\mathcal O\in\opens$.
\end{enumerate*}
Write $\opensne=\{O\in\opens\mid O\neq\varnothing\}$.
We may call $O\in\opensne$ a \deffont{quorum}.
Note that we assumed $\points\neq\varnothing$ so that also $\opensne\neq\varnothing$ (i.e. some quorum exists).
\end{defn}

\begin{figure}
$$
\begin{array}{r@{\ }l@{\qquad}r@{\ }l@{\qquad}r@{\ }l@{\qquad}r@{\ }l}
\everywhere f\tvequal &\bigwedge_{p\in \points} f(p) 
&
\somewhere f\tvequal &\bigvee_{p\in \points} f(p) 
\\
\Quorum f\tvequal &\bigvee_{O\in\opensne} \bigwedge_{p\in O} f(p) 
&
\Coquorum f \tvequal&\bigwedge_{O\in\opensne} \bigvee_{p\in O} f(p) . 
\end{array}
$$ 
Above, $\everywhere,\somewhere,\Quorum,\Contraquorum:(\points\to\THREE)\to\THREE$ and we read $\everywhere f$ as \emph{everywhere}-$f$, $\somewhere f$ as \emph{somewhere}-$f$, 
$\Quorum f$ as \emph{quorum}-$f$, and $\Coquorum f$ as \emph{contraquorum}-$f$. 
\caption{Modalities (Definition~\ref{defn.semitopology.logic})}
\label{fig.sl.modalities}
$$
\begin{array}{r@{\ }l@{\qquad}r@{\ }l}
\tneg(f \tand f')\tvequal&(\tneg f\tor\tneg f')
&
\tneg(f \tor f')\tvequal&(\tneg f\tand \tneg f')
\\
\tneg\someone(\tneg f)\tvequal&\everyone f
&
\tneg\everyone(\tneg f)\tvequal&\someone f
\\
\tneg\Contraquorum (\tneg f)\tvequal&\Quorum f
&
\tneg\Quorum(\tneg f)\tvequal&\Contraquorum f
\end{array}
$$
Above: $f,f':\points\to\THREE$\ and\ $\tand,\tor:(\THREE\to\THREE)^2\to(\THREE\to\THREE)$ are defined pointwise as per Definition~\ref{defn.THREE}(\ref{item.THREE.pointwise});\ $\everywhere,\somewhere,\Quorum,\Contraquorum:(\points\to\THREE)\to\THREE$;\ innermost $\tneg$ has type $(\points\to\THREE)\to(\points\to\THREE)$ (it is defined on $f$ and $f'$ pointwise), and outermost $\tneg$ has type $\THREE\to\THREE$. 
\caption{De Morgan dualities (Lemma~\ref{lemm.dmc})}
\label{fig.dmc}
\end{figure}

A semitopology naturally gives rise to a modal logical structure, as follows:
\begin{defn}
\label{defn.semitopology.logic}
Suppose $(\points,\opens)$ is a semitopology and $f:\points\to\THREE$ is a function, which we can think of as a unary predicate (a three-valued one, taking truth-values in $\THREE$).
Define functions 
$$
\everywhere,\somewhere,\Quorum,\Contraquorum:(\points\to\THREE)\to\THREE
$$
--- where we read 
$\everywhere f$ as `\deffont{everywhere}-$f$', $\somewhere f$ as `\deffont{somewhere}-$f$', 
$\Quorum f$ as `\deffont{quorum}-$f$', and $\Coquorum f$ as `\deffont{contraquorum}-$f$' 
--- 
as per Figure~\ref{fig.sl.modalities}.
\end{defn}

\begin{lemm}
\label{lemm.dmc}
We note some \emph{de Morgan} 
dualities: 
\begin{enumerate*}
\item\label{item.dmc.1}
Suppose $\points$ is a set and $f,f':\points\to\THREE$. 
Then we have the equivalences in Figure~\ref{fig.dmc}.
\item\label{item.dmc.2}
Suppose $\f{tv}\in\THREE$.
Then $\tneg\modT\tneg\f{tv}\tvequal\modTB\f{tv}$ and $\tneg\modTB\tneg\f{tv}\tvequal\modT\f{tv}$.
\end{enumerate*}
\end{lemm}
\begin{proof}
By routine arguments on lattices and truth-tables.
\end{proof}

\begin{rmrk}
\label{rmrk.contraquorum}
Semitopologies~\cite{gabbay:semdca,gabbay:semtad} were designed to generalise the quorum systems like Example~\ref{xmpl.vote} (expositions are in~\cite{DBLP:journals/dc/MalkhiR98,DBLP:journals/dc/AlposCTZ24}):
\emph{quorum} corresponds to \emph{nonempty open set}; \emph{coquorum} (the set complement to a quorum) corresponds to \emph{non-total closed set}; and \emph{contraquorum} or \emph{blocking set} (a set that intersects every quorum) corresponds to \emph{dense set}. 

There are some nice benefits to working with a semitopology as per Definition~\ref{defn.semitopology}, rather than with a concrete structure as per Example~\ref{xmpl.vote}:
\begin{enumerate*}
\item
`$O\in\opensne$' is shorter, and more mathematically rich, than `a set $O$ having $2f\plus 1$ elements'.
\item 
It notes something that is arguably obvious in retrospect, but which does not appear to have been noted before semitopologies: that certain very basic concepts in the distributed systems literature --- including quorums and blocking sets --- are topological in nature --- corresponding e.g. to nonempty open sets and dense sets.\footnote{While the observation `quorum system'=`nonempty open set of a semitopology' appears to be novel, there is pedigree for making other kinds of connections between (algebraic) topology and distributed systems.  We discuss this, with references, in the Conclusions (Section~\ref{sect.conclusions}).} 
\item
As Definition~\ref{defn.semitopology.logic} and Lemma~\ref{lemm.dmc} suggest, and the rest of this paper substantiates, we can exploit the connection between topology and logic to reason axiomatically on distributed algorithms.
\end{enumerate*}
\end{rmrk}

\begin{rmrk}
$\modT$ and $\modTB$ commute with all monotone connectives: for $\tf M\in\{\modT,\modTB\}$ and $\f{tv},\f{tv}'\in\THREE$ we have $\tf M(\f{tv}\tand\f{tv}')\tvequal((\tf M\f{tv})\tand\tf M\f{tv}')$ and similarly for $\tor$, $\someone$, $\everyone$, $\Quorum$, and $\Contraquorum$.

The reader can also easily check facts like $\modT\modTB\f{tv}=\modT\f{tv}=\modTB\modT\f{tv}$, and $\modT\modB\f{tv}=\modB\f{tv}$, and so on.

We may rewrite using these properties without comment henceforth.  
\end{rmrk} 

We warm up with an easy lemma that connects $\everyone$ with $\quorum$:
\begin{lemm}
\label{lemm.everyone.and.quorum}
Suppose $(\points,\opens)$ is a semitopology and $f,f':\points\to\THREE$.
Then 
$$
(\everyone f\tand \quorum f') \leq \quorum(f\tand f').
$$
(Recall that $\tvF<\tvB<\tvT$ in $\THREE$ as a lattice.  The $\leq$ above refers to this lattice ordering.)

As corollaries:
\begin{enumerate*}
\item\label{item.everyone.and.quorum.1}
$\ment\everyone f\tand\quorum f'$ implies $\ment\quorum (f\tand f')$.
\item\label{item.everyone.and.quorum.2}
$\ment\everyone f\tand\quorum(\modTF{\circ} f)$ implies $\ment\modT\quorum f$.
\end{enumerate*}
\end{lemm}
\begin{proof}
By routine lattice calculations from Definition~\ref{defn.semitopology.logic}.
Or, using topological language: if every $p\in\points$ satisfies $f(p)$, and if there exists $O'\in\opensne$ such that every $p\in O'$ satisfies $f'(p)$, then clearly every $p\in O'\in\opensne$ satisfies $f(p)\tand f'(p)$.
The first corollary follows by routine reasoning, and the second follows taking $f'=\modTF\circ f$.
\end{proof}

A characteristic fact connecting $\Quorum$, $\Contraquorum$, and $\someone$ is this:
\begin{lemm}
\label{lemm.semi.char}
Suppose $(\points,\opens)$ is a semitopology and $f,f':\points\to\THREE$.
Then 
$$
(\Quorum f\tand\Contraquorum f')\leq \someone(f\tand f') .
$$ 
As corollaries:
\begin{enumerate*}
\item\label{item.semi.char.1}
$\ment\Quorum f\tand\Contraquorum f'$ implies $\ment\someone(f\tand f')$.
\item\label{item.semi.char.2}
$\ment\Contraquorum f'$ implies $\ment\someone f'$.
\item\label{item.semi.char.TF}
$\ment\quorum(\modTF\circ f')$ and $\ment\Contraquorum f'$ implies $\ment\modT\someone f'$ (equivalently: $\ment\someone\modT f'$).
\end{enumerate*}
\end{lemm}
\begin{proof}
By routine lattice calculations from Definition~\ref{defn.semitopology.logic}.
Or, using topological language: 
if we have an open set of $f$s and a dense set of $f'$s, then there must be some point with both $f$ and $f'$.
Part~\ref{item.semi.char.1} follows; part~\ref{item.semi.char.2} follows from part~\ref{item.semi.char.1} setting $f=\lambda p.\tvT$ and recalling from Definition~\ref{defn.semitopology} that $\opensne\neq\varnothing$ so that the universal quantification in $\Coquorum f$ (Definition~\ref{defn.semitopology.logic})
cannot be vacuously satisfied; and part~\ref{item.semi.char.TF} follows from part~\ref{item.semi.char.1} taking $f=\modTF\circ f'$. 
\end{proof}

\begin{rmrk}
\label{rmrk.nonempty.P}
In Definition~\ref{defn.semitopology}(\ref{item.semitopology.points}) we assume that the set of points is nonempty, disallowing the \emph{empty semitopology} $(\varnothing,\{\varnothing\})$.
In~\cite[Definition~1.2.2]{gabbay:semdca} and~\cite[Definition~1.1.2]{gabbay:semtad} we do not impose this condition, i.e. we allow the empty semitopology.
There are two reasons for this difference:
\begin{enumerate*}
\item
In~\cite{gabbay:semdca,gabbay:semtad} we are studying semitopologies \emph{in the abstract}, so we want to be as general as possible.
Here, we use semitopologies specifically to model quorum systems. 

A distributed system with no participants is just not interesting for our purposes here; there is literally nothing to talk about and in particular this semitopology would have no quorums (nonempty open sets).
In context, this possibility would be pathological and possibly confusing, so we outlaw it by assuming that the set of participants is nonempty.
\item
Related to the previous point, but more technically, admitting the empty semitopology would mean that $\Coquorum(\lambda p.\tvF)$ could equal $\tvT$.
Mathematically that is perfectly respectable, but it would slightly complicate the statement of Lemma~\ref{lemm.semi.char}(\ref{item.semi.char.2}); we would need to explicitly add a condition `for nonempty semitopology'.\footnote{Thanks to Jan Mas Rovira for noting this subtlety.}
It is simpler and neater to just insist that all semitopologies are nonempty.
\end{enumerate*} 
\end{rmrk}

\begin{rmrk}
In Definition~\ref{defn.semitopology} we define semitopologies, which (as per Remark~\ref{rmrk.contraquorum}) abstract the notion of \emph{quorum system}.
Definition~\ref{defn.semitopology}(\ref{item.semitopology.opens}) insists that the open sets of a semitopology be closed under unions, but in fact we will not use that condition in the rest of this paper!
It is there because it is inherited from the ambient semitopological mathematics; for the purposes of \emph{this} paper, we could drop it and the mathematics would work just as well.

To couch this in fancy jargon: we could work with just a basis of the semitopology, where a \deffont{basis} of a semitopology $(\points,\opens)$ is a subset $\mathcal O\subseteq\opens$ that generates all of $\opens$ by arbitrary set unions.

However, in a deeper sense closure under unions is very much still present:
\begin{enumerate*}
\item
In the literature, quorum systems used for voting, Bracha Broadcast, and Crusader Agreement are traditionally what we can call \emph{threshold-based}; meaning that they have the form `a quorum is a set with more than $t$ elements' for some threshold $t$, where usually $t=n*f+1$ for parameters $n$ and $f$ (see Example~\ref{xmpl.vote}, Subsection~\ref{subsect.3-twined}, and Example~\ref{xmpl.n-wtined}(\ref{item.nf+1})).

Threshold-based quorum systems \emph{are} closed under arbitrary unions.
Indeed, they satisfy the stronger property of being \emph{up-closed} (if $O\in\opens$ and $O\subseteq O'\subseteq\points$ then $O'\in\opens$).
Thus the semitopologies of Definition~\ref{defn.semitopology} as written are a \emph{de facto} generalisation of the quorum systems as used in the literature relevant to the examples in this paper. 
\item
Related to the previous point, the form of our axioms and correctness properties will be such that they are functionally insensitive to the difference between a full semitopology and a basis of that semitopology.
Specifically, all conditions on the size of a set of participants are positive, and we will not see any negative condition in an axiom or correctness property that a set of participants should `not be too big'.\footnote{This is also arguably part of why quorum systems in the literature are up-closed; since up-closure is functionally invisible to the protocols, a threshold quorum system is perfectly sufficient and it is unintuitive, and not worth the bother, to insist on quorums having `exactly $t$ elements'.}
In this sense, the reason that we do not use the closure under unions condition from Definition~\ref{defn.semitopology}(\ref{item.semitopology.opens}) explicitly, is because it is so fundamental that it is actually structurally built in to the form of the axioms and correctness properties.

This positivity criterion may be a useful static (i.e. syntactic) criterion for what makes a good axiom system, but we leave thinking about that to future work (see also Subsection~\ref{subsect.implementation}).
\end{enumerate*}
To sum up, closure under unions is not explicitly used in this paper, but it lurks in the background: in the form of concrete quorum systems used in the literature; in the ambient semitopological metatheory; in the fact that threshold-based quorum systems satisfy this property anyway; \emph{and} also (I would argue) it is implicit in the sense that our axioms and correctness properties are constituted precisely to work as well with a full semitopology as with a basis of that semitopology, and this seems likely to be an important well-behavedness criterion for protocols.
\end{rmrk}

\subsection{3-twined semitopologies}
\label{subsect.3-twined}

\subsubsection{Definition, discussion, and examples}

\begin{defn}
\label{defn.3twined}
For $n\geq 0$, call a semitopology $(\points,\opens)$ \deffont{$n$-twined} when any $n$ nonempty open sets have a nonempty intersection~\cite[Section~5]{gabbay:decasd}.
That is: if $O_1,\dots,O_n\in\opensne$ then $\bigcap_{1\leq i\leq n} O_i\neq\varnothing$.

In this paper we will care most about 3-twined semitopologies; semitopologies such that any three nonempty open sets have a nonempty intersection. 
\end{defn}

\begin{rmrk}
Indexing the open sets from $1$ in Definition~\ref{defn.3twined} even though we only insist that $n\geq 0$, is not a typo.
By convention, $\bigcap\varnothing=\points$, so a semitopology is $0$-twined just when $\points$ is nonempty; as per Remark~\ref{rmrk.nonempty.P} $\points$ is assumed nonempty, so every (nonempty) semitopology is $0$-twined.
The reader can also check that every semitopology is $1$-twined.

So the $n$-twined condition starts getting more structured from $2$-twined onwards.

The terminology `3-twined' follows the semitopologies literature: it generalises the notion of \emph{intertwined} space considered in~\cite{gabbay:semdca,gabbay:semtad} (intertwined = $2$-twined) and see also the Section on $n$-twined semitopologies in~\cite{gabbay:decasd}.
However, the general concept is well-known in the literature, albeit not phrased systematically and in (semi)topological terms. 
For example, being 3-twined is called the $Q^3$ property in~\cite[Definition 2]{DBLP:journals/dc/AlposCTZ24}. %
\end{rmrk}

\begin{rmrk}
We take a moment to motivate 3-twinedness: why is this of such interest to distributed systems?
Note that:
\begin{itemize*}
\item
A typical step in a distributed protocol has the general form: 
\begin{quote}
``Wait until you see a quorum $Q$ (= nonempty open set) of participants perform some action $A$, and then perform action $B$''.
\end{quote}
For example, the \rulefont{BrDeliver?/!} and \rulefont{BrReady?/!} axiom-pairs in Figure~\ref{fig.bb} have \emph{precisely} this form; and other axioms are variations on it.
\item
A typical correctness assumption is that there exists a quorum $C$ of correct (= honest / non-faulty) participants (see for example \rulefont{BrCorrect} and \rulefont{BrCorrect'}).
\end{itemize*}
Now consider two participants $p$ and $p'$ running some protocol $X$, and suppose that:
\begin{enumerate*}
\item
$p$ sees a quorum $Q$ of participants perform action $A$ and so according to $X$ is ready to perform action $B$, and 
\item
$p'$ sees a quorum $Q'$ of participants perform action $A'$ and so according to $X$ is ready to perform action $B'$. 
\end{enumerate*}
If the semitopology is 3-twined, then there exists some $q\in Q\cap Q'\cap C$. 
This means that $p$ and $p'$ know that there exists some $q$ who is honest and has performed actions $A$ and $A'$; they do not know who $q$ is, but they know that $q$ exists.

This can give $p$ and $p'$ confidence to progress to perform actions $B$ and $B'$ respectively, safe in the knowledge that $A$ and $A'$ are compatible, in the sense that there was at least one honest participant who was willing to do both.

If this sounds trivial, it is not.
For example $A$ might be `vote for value $v$' and $A'$ might be `vote for $v'$', and perhaps the protocol insists that (honest) participants should only vote for one possibility.
Then in this case, $p$ and $p'$ know that $v=v'$ --- even though they do not know who the honest participants are, and even though $p$ and $p'$ may not even know about each other. 
In a bit more detail, 3-twinedness gives us this intuitive property:
\begin{quote}
every honest participant knows that some honest participant exists in any pairwise quorum intersection, 
\end{quote}
and as per our example %
above, this in and of itself may be enough confidence for individual participants like $p$ and $p'$ (even without speaking to one another) to progress with confidence just on the basis of observing a relevant quorum. 

This in a nutshell is the content of Theorem~\ref{thrm.3twined.logic}, Corollary~\ref{corr.3twined.cologic}, and axiom \rulefont{3twined} in Figure~\ref{fig.simple}, and furthermore, Remark~\ref{rmrk.3twined.logic.back} proves that this intuition is (in a sense made mathematical in that Remark) not only sound but also complete.
\end{rmrk}

The semitopology in Example~\ref{xmpl.vote} is the canonical example of a 3-twined semitopology, but trivially an $n$-twined semitopology is $n'$-twined for every $n'\leq n$, and examples of $n$-twined semitopologies for arbitrary $n$ are very easy to generate, for example as follows:
\begin{xmpl}
\label{xmpl.n-wtined}
We indicate two schemas by which the reader can generate $n$-twined semitopologies.
For $n\geq 3$, any of these could be used in the development to follow, with no change to the proofs (since all we will require is for our semitopology to be 3-twined):
\begin{enumerate*}
\item\label{item.nf+1}
Suppose $f\geq 0$ and $n\geq 1$ and let $\points$ be any set with cardinality $n{*}f\plus 1$.
Endow $\points$ with a semitopological structure by letting $\opens$ be the set containing the empty set, and containing any set with cardinality at least $(n\minus 1){*}f+1$.
Then the reader can easily check that $(\points,\opens)$ is an $n$-twined semitopology.
\item
Fix a logical theory $\Phi$ in first-order logic\footnote{A \emph{logical theory} is a set of closed predicates (having no free variables; they can mention bound variables), which we call \emph{axioms}.} and fix a first-order logic predicate $\phi(x,y)$ with (at most) two free variables $x$ and $y$, and assume the following property: 
$$
\Phi\cent \mathit{n\text{-}twined}_\phi
\quad\text{where}\quad
\mathit{n\text{-}twined}_\phi=\Forall{y_1,\dots,y_n}\Exists{x}(\phi(x,y_1)\land\dots\land\phi(x,y_n)).
$$
Here $\cent$ is the standard symbol for derivability, so $\Phi\cent\mathit{n\text{-}twined}_\phi$ means that we can prove that for every $y_1,\dots,y_n$ there exists $x$ such that $\phi(x,y_1)$,\dots,$\phi(x,y_n)$ are all simultaneously true.
By Soundness of first-order logic, it follows that $\mathit{n\text{-}twined}_\phi$ is valid in every model of $\Phi$.

Now fix some model $M$ of $\Phi$.
This means that every axiom in $\Phi$ is valid in $M$, so that by Soundness and our assumption that $\Phi\cent\mathit{n\text{-}twined}_\phi$, also $M\ment\mathit{n\text{-}twined}_\phi$.
Here $\ment$ is the standard symbol for model-theoretic validity, so $M\ment\mathit{n\text{-}twined}_\phi$ means that $\mathit{n\text{-}twined}_\phi$ \emph{is actually valid} of $M$.

Give (the underlying set of) $M$ a semitopological structure by letting open sets be the empty set, the whole set, or arbitrary unions of sets of the form $\{x\in M \mid {M\ment\phi(x,y)}\}$ where $y\in M$.
By construction this makes $M$ into an $n$-twined semitopology, since $\mathit{n\text{-}twined}_\phi$ ensures that any $n$ nonempty open sets \emph{actually do} intersect at some $x$.

Now this example might seem like just an absurdly abstract way to generate $n$-twined semitopologies, and in a sense it is, but in another sense it is much more than that: it is \emph{the} way to generate $n$-twined semitopologies.

A choice of $\Phi$ and $\phi$ such that $\Phi\cent\mathit{n\text{-}twined}_\phi$, specifies a class of $n$-twined semitopologies, and conversely, to the extent that any practical construction of $n$-twined semitopologies can be expressed in first-order logic, \emph{every} such construction %
can be so expressed.\footnote{In principle we could imagine using more powerful logics, but every quorum system that I have seen in the literature so far is clearly a first-order construction.} 

So the difference between this example and (say) the grid quorum systems in~\cite[Section~4]{senn2025asymmetricgridquorumsystems} is just how specific we are about picking $\Phi$ and $\phi$. 
\end{enumerate*}
Because we are using an axiomatic approach, we do not need to be any more specific about the underlying model than the axioms require.
In what follows, we can and we will just reason assuming an arbitrary 3-twined semitopology.
This is all the detail that the axioms require, and so that will be all that we need.
\end{xmpl}

\subsubsection{The fundamental logical characterisation of being 3-twined, and its corollaries}

Reading Definition~\ref{defn.3twined} literally, we can render it as $(\Quorum f \tand \Quorum f' \tand \Quorum f'')\leq \someone(f\tand f'\tand f'')$, for any $f,f',f'':\points{\to}\THREE$.
However, this is not always the most convenient logical form of the property. 
Arguably, the fundamental logical property of a 3-twined semitopology is this:
\begin{thrm}
\label{thrm.3twined.logic}
Suppose $(\points,\opens)$ is 3-twined and $f,f':\points{\to}\THREE$.
Then
$$
(\Quorum f \tand \Quorum f')\leq \Contraquorum(f\tand f').
$$
As a corollary 
$$
\ment\Quorum f \tand \Quorum f'\quad\text{implies}\quad \ment\Contraquorum(f\tand f').
$$
(A converse to this result also holds: see Remark~\ref{rmrk.3twined.logic.back}.)
\end{thrm}
\begin{proof}
$\Quorum f\tand\Quorum f'=\tvT$ implies that there exist $O,O'\in\opensne$ on which $f$ and $f'$ respectively take the value $\tvT$.
By the 3-twined property, $O\cap O'$ intersects every other $O''\in\opensne$, thus $\Contraquorum(f\tand f')=\tvT$.

$\Quorum f\tand\Quorum f'=\tvB$ implies that there exist $O,O'\in\opensne$ on which $f$ and $f'$ respectively take values in $\{\tvT,\tvB\}$.
By the 3-twined property $O\cap O'$ intersects every other $O''\in\opensne$, thus $\Contraquorum(f\tand f')\in\{\tvT,\tvB\}$.
\end{proof}

Corollary~\ref{corr.3twined.cologic} is a (rather elegant) dual rephrasing of Theorem~\ref{thrm.3twined.logic} and will be useful later:
\begin{corr}
\label{corr.3twined.cologic}
Suppose $(\points,\opens)$ is 3-twined and $f,f':\points\to\THREE$.
Then
$$
\quorum(f\tor f') \leq (\contraquorum f \tor \contraquorum f').
$$
As a corollary, 
$$
\ment\quorum(f\tor f') \quad\text{implies}\quad \ment\contraquorum f \tor \contraquorum f'.
$$
\end{corr}
\begin{proof}
By Theorem~\ref{thrm.3twined.logic} 
$(\quorum\tneg f \tand \quorum\tneg f')\leq \contraquorum(\tneg f\tand \tneg f')$.
By Proposition~\ref{prop.tand.tor}(\ref{item.para.de.Morgan}) $\tneg\contraquorum(\tneg f\tand \tneg f')\leq \tneg(\quorum\tneg f \tand \quorum\tneg f')$. 
Simplifying using Lemma~\ref{lemm.dmc}(\ref{item.dmc.1}) we obtain
$\quorum(f\tor f')\leq (\contraquorum f \tor \contraquorum f')$ as required. 
\end{proof}

\begin{rmrk}
It might be helpful to sum up some highlights of the maths so far.
Suppose $(\points,\opens)$ is a semitopology and $f,f':\points\to\THREE$.

Suppose further that $\ment\quorum(\modTF\circ f)$, meaning that $f$ is correct (= $\tvT$ or $\tvF$, not $\tvB$) on a quorum (= nonempty open set) of participants.
(It is a typical failure assumption in the distributed systems literature that there exists a quorum of correct participants.)

Then we have seen that:
\begin{enumerate*}
\item
$\ment\everyone f$ implies $\ment\modT\quorum f$, by Lemma~\ref{lemm.everyone.and.quorum}(\ref{item.everyone.and.quorum.2}).
\item
If $(\points,\opens)$ is 3-twined then $\ment\quorum f$ implies $\ment\contraquorum(\modT{\circ} f)$ and equivalently $\ment\modT\contraquorum f$, by Theorem~\ref{thrm.3twined.logic}.
\item
$\ment\contraquorum f$ implies $\ment\someone(\modT{\circ} f)$ and equivalently $\ment\modT\someone f$, by Lemma~\ref{lemm.semi.char}(\ref{item.semi.char.TF}).
\item
$\ment\quorum f\tand\contraquorum(\modT{\circ} f)$ implies $\ment\modT\someone f$, from Lemma~\ref{lemm.semi.char}(\ref{item.semi.char.1}) (taking $f'$ in that Lemma to be $\modT{\circ}f$).
\item
If $(\points,\opens)$ is 3-twined then $\ment\Quorum f \tand \Quorum f'$ implies $\ment\Contraquorum(f\tand f')$ and $\ment\quorum(f\tor f')$ implies $\ment\contraquorum f \tor \contraquorum f'$, by Theorem~\ref{thrm.3twined.logic} and Corollary~\ref{corr.3twined.cologic}.
\end{enumerate*}
\end{rmrk}

\begin{rmrk}
\label{rmrk.3twined.logic.back}
For completeness we mention that the reverse implication in Theorem~\ref{thrm.3twined.logic} also holds (though we will not need it here): if for every $f,f':\points\to\THREE$ we have $(\Quorum f \tand \Quorum f')\leq \Contraquorum(f\tand f')$, then $(\points,\opens)$ is 3-twined.  

The argument is as follows:
Suppose we have three nonempty open sets $O,O',O''\in\opensne$.
Following Definition~\ref{defn.3twined} we must show that $O\cap O'\cap O''\neq\varnothing$, and our assumption is that for any functions $f,f':\points\to\THREE$ we have $(\Quorum f\tand\Quorum f')\leq\Contraquorum(f\tand f')$.
We choose $f$ to map $p\in O$ to $\tvT\in\THREE$ and $p\not\in O$ to $\tvB\in\THREE$, and we choose $f'$ to map $p\in O'$ to $\tvT$ and $p\not\in O'$ to $\tvB$.\footnote{Why $\tvB$ and not $\tvF$?  We could just as well use $\tvF$, but the benefit of using $\tvB$ is that $f$ and $f'$ are also \emph{continuous} to $\THREE$ with its natural semitopological structure as developed in~\cite{gabbay:semdca,gabbay:semtad}.  In this paper we do not develop continuity, but in a wider context it is nice to have an `if-and-only-if' not just with respect to the set of all $f,f':\points\to\THREE$, but also specifically for the continuous ones.} 
By the construction of $\Quorum$ in Definition~\ref{defn.semitopology.logic}, $\Quorum f=\tvT=\Quorum f'$, so by our assumption $\Contraquorum(f\tand f')=\tvT$.
By the construction of $\Contraquorum$ in Definition~\ref{defn.semitopology.logic}, it must be that $O\cap O'\cap O'''$ is nonempty for \emph{every} $O'''\in\opensne$, and in particular $O\cap O'\cap O''$ is nonempty as required.
\end{rmrk}

\subsection{A simple voting protocol}

We are now ready to consider our first (very simple) protocol:
\begin{defn}
\label{defn.voting}
Call a tuple $\mu=((\points,\opens),\fvote,\fdecide)$, where 
\begin{itemize*}
\item
$(\points,\opens)$ is a(n arbitrary) semitopology and 
\item
$\fvote,\fdecide:\points\to\THREE$ are functions, 
\end{itemize*}
a \deffont{model}.
We call a model $\mu$ a \deffont{model} (of \ThyVote) when the axioms in Figure~\ref{fig.simple} are valid in $\mu$.
\end{defn}

\begin{rmrk}
We take a minute to unpack what the notation in Figure~\ref{fig.simple} is intended to mean:
\begin{enumerate*}
\item
$\fvote(p)=\tvT$ indicates that $p$ voted for $\tvT$.
$\fvote(p)=\tvF$ indicates that $p$ voted for $\tvF$.
$\fvote(p)=\tvB$ indicates that $p$ (is hostile because it) sent vote messages for \emph{both} $\tvT$ and $\tvF$.\footnote{%
There is no notion here of `not voted yet', because there is no notion of time!
$\fvote$ just records the votes that a participant casts, and our failure assumptions assume that \emph{some} vote is made, as per Remark~\ref{rmrk.failure.assumptions.voting}.
Different assumptions would yield different axioms.}
\item
$\fdecide(p)=\tvT$ indicates that $p$ observed a quorum of vote-$\tvT$ messages -- some of which may be from hostile participants; $p$ cannot know.
Similarly for $\fdecide(p)=\tvF$. 
$\fdecide(p)=\tvB$ indicates that $p$ observed no quorum of messages either way.
\item
$\fvote$ and $\fdecide$ may look like messages (`vote', `observe'), but they are functions of type $\points\to\THREE$ (the reader used to higher-order logic may wish to call them \emph{predicates}).

The intuitive meaning of $\fvote(p)$ and $\fdecide(p)$ is `participant $p$ voted/observed during a run of the protocol'.\footnote{Normally we would vote for or observe a value, but in this simple protocol we leave the values out.}
But the mathematical meaning is just applying a function to a $p\in\points$ to get a truth-value in $\THREE$. 
\end{enumerate*}
\end{rmrk}

\begin{rmrk}
We discuss how the axioms in Figure~\ref{fig.simple} correspond to the protocol in Example~\ref{xmpl.vote}:
\begin{enumerate*}
\item
\rulefont{Observe?} is a \deffont{backward rule}: it expresses that \emph{if} a participant $p$ observes $\tvT$, then $p$ must have received a quorum of vote-$\tvT$ messages.

We use weak implication $\tnotor$ (not strong implication $\timpc$) in \rulefont{Observe?} because a quorum may include some dishonest (hostile) participants who sent dishonest messages; i.e. vote-$\tvT$ to some participants and vote-$\tvF$ to others.
We represent the combination of vote-$\tvT$ and vote-$\tvF$ messages from a hostile participant $q$ by setting $\fvote(q)=\tvB$. 
Thus, if $\fdecide(p)=\tvT$ it does not follow that $\Quorum\fvote=\tvT$; it only follows that $\Quorum\fvote\in\{\tvT,\tvB\}$.
This highlights that $\fvote(p)$ in our declarative model is a single \emph{truth-value}, not a \emph{message} or a set of messages.
\item
\rulefont{Observe\tneg?} just repeats \rulefont{Observe?} for the case that $p$ receives a quorum of vote-$\tvF$.
\item
\rulefont{Observe!}/\rulefont{Observe\tneg!} are \deffont{forward rules}: they express that if an honest quorum votes for $\tvT$/$\tvF$, then every $p$ observes $\tvT$/$\tvF$.

Why did we use strong implication here, not weak implication?
We could use either, but we make a design choice: since (in our protocol) participants do not report their $\fdecide$ to anyone else, there is no useful sense in which lying would be meaningful.
So we use the strong implication.
\item
\rulefont{Correct} asserts that some quorum of participants are honest: they vote $\tvT$ or $\tvF$, and not $\tvB$.
\item
\rulefont{3twined} expresses the characteristic intersection property of a 3-twined semitopology, as per Theorem~\ref{thrm.3twined.logic}.
For this protocol we will use it once, setting $f=\fvote$ and $f'=\tneg\fvote$.
\end{enumerate*}
\end{rmrk}

\begin{rmrk}
At the level of $\fvote$, $\tvB$ denotes a hostile participant sending mixed votes to different participants. 
At the level of $\fdecide$, $\tvB$ just indicates that a participant has observed no quorum of participants voting for $\tvT$ or voting for $\tvF$.

This use of $\tvB$ to mean two things, depending on the context, is a feature, not a bug.
It illustrates a fact about truth-values in $\THREE$: they are just data, which we -- via our choice of axioms and intended interpretation -- can interpret as we choose and as is mathematically convenient.\footnote{An example from real life.  If my wife asks me ``You didn't forget to put the rubbish out, did you?'' and I return a value ``Yes'', does this mean ``Yes, I did forget!'' or ``Yes, I did put the rubbish out!''?  The ``Yes'' here is raw data; its meaning exists in the (hopefully shared) understanding of the people having the conversation.

In the same way, $\tvB$ is just data, and is meaningless until we -- \emph{we} -- use it in axioms and interpret those axioms.
} 
\end{rmrk}

\begin{rmrk}
Note in Figure~\ref{fig.simple} that \rulefont{Observe?} has a dual \rulefont{Observe\tneg?}, and \rulefont{Observe!} has a dual \rulefont{Observe\tneg!} -- but \rulefont{Correct} has no dual axiom \rulefont{Correct\tneg}.
Why is this so?
Because \rulefont{Correct} is self-dual: it is a fact of the truth-tables in Figure~\ref{fig.3} that $\modTF\circ\fvote = \modTF\circ\tneg\circ\fvote$, so that $\Quorum(\modTF\circ\fvote)\tvequal\Quorum(\modTF\circ\tneg\circ\fvote)$.
\end{rmrk}

\begin{lemm}
$\ment \someone\fdecide\tnotor\Quorum\fvote$,\ $\ment \Quorum\fvote\timpc \everyone\fdecide$,\  
$\ment \someone\tneg\fdecide\tnotor\Quorum\tneg\fvote$ and
$\ment \Quorum\tneg\fvote\timpc \everyone\tneg\fdecide$. 
\end{lemm}
\begin{proof}
The reasoning is very easy and we prove just the first two parts.
\begin{enumerate*}
\item
By \weakmodusponens it suffices to show that $\ment\modT\someone\fdecide$ implies $\ment\Quorum\fvote$.
So suppose $\ment\modT\someone\fdecide$.
Using the meaning of $\someone$ from Definition~\ref{defn.semitopology.logic}, there exists $p\in\points$ such that $\fdecide(p)=\tvT$, and so $\ment\modT\fdecide(p)$.
Then using \weakmodusponensnoref and \rulefont{Observe?} we have $\ment\Quorum\fvote$.
\item
By \strongmodusponens and the meaning of $\everyone$, it suffices to show that $\ment\modT\Quorum\fvote$ implies $\ment\modT\fdecide(p)$ for every $p\in\points$.
But this is immediate from \rulefont{Observe!}.
\qedhere\end{enumerate*}
\end{proof}

We can now prove it is impossible for one participant to observe $\tvT$ and another to observe $\tvF$:
\begin{prop}[Agreement]
\label{prop.vote.correctness}
If $\mu$ is a model of \ThyVote
then 
$\nment\someone\modT\fdecide \tand \someone\modT\tneg\fdecide$.
\end{prop}
\begin{proof}
Suppose $\ment\someone\modT\fdecide\tand\someone\modT\tneg\fdecide$. 
Then there exist $p,p'\in\points$ such that $\fdecide(p)=\tvT$ and $\fdecide(p')=\tvF$.
Using \rulefont{Observe?} and \rulefont{Observe\tneg?} $\ment\Quorum\fvote\tand\Quorum\tneg\fvote$. 
By \rulefont{3twined} $\ment\Contraquorum(\fvote\tand\tneg\fvote)$, so using 
Proposition~\ref{prop.tand.tor}(\ref{item.para.para}) $\ment\Contraquorum\modB\fvote$.
Using \rulefont{Correct} and Lemma~\ref{lemm.semi.char}(\ref{item.semi.char.1}) $\ment\someone(\modB\fvote\tand\modTF\fvote)$.
But $(\modB\fvote\tand\modTF\fvote) \tvequal \tvF$, and $\ment\someone\tvF$ is impossible.
\end{proof}

\begin{rmrk}
A more traditional proof for Proposition~\ref{prop.vote.correctness} would work concretely with the set of quorums from Example~\ref{xmpl.vote}, and prove Agreement by counting and calculations (see for example Lemmas~1 to~4 used in the proof of Theorem~1 from~\cite[page~135]{bracha:asybap}). 

Proposition~\ref{prop.vote.correctness} is not that proof: the semitopology is abstract, except for satisfying the axioms in Figure~\ref{fig.simple}.
The reasoning is logical in flavour and follows the structure of the axioms.
This is simple, general, and it captures the logical essence of what the concrete counting proofs are actually doing.
\end{rmrk}

\begin{rmrk}
The failure model is encoded in the axioms -- typically this shows up in the strength of the implications we use.
For example: our failure model in Remark~\ref{rmrk.failure.assumptions.voting} assumed a reliable network, which is reflected by the use of \emph{strong implication} in the forward rules \rulefont{Observe!} and \rulefont{Observe\tneg!} in Figure~\ref{fig.simple}.

To model an unreliable network, in which messages sent might not arrive, we would weaken \rulefont{Observe!} to $\Quorum\fvote\tnotor\fdecide$, and similarly for \rulefont{Observe\tneg!}.
This reflects that a quorum of honest participants vote for $\tvT$, but the network fails to transmit their messages so a participant $p$ (even if honest) may remain unable to observe and so set $\fdecide(p)=\tvB$.
\end{rmrk}

\section{The modal logic}
\label{sect.logic}

The above concludes our discussion of the simple voting protocol.
We are now ready to start studying Bracha Broadcast.
We will need two more ingredients: 
\begin{enumerate*}
\item
Our voting example was binary (vote-$\tvT$ or vote-$\tvF$).
For Bracha Broadcast we assume an arbitrary nonempty set of values $\Val$, over which we will need three-valued quantifiers like unique existence $\texiunique:\Val\to\THREE$.
We study this in Subsection~\ref{subsect.existence}.
\item
Our voting example worked directly with semantics; we made assertions directly about functions like $\fvote:\points\to\THREE$.
For Bracha Broadcast it is better to build the syntax and semantics of a (simple) logic; we will have predicate-symbols like $\tf{echo}$, which will be interpreted as functions like $\varsigma(\tf{echo}):\points\to\Val\to\THREE$.
We study this in Subsection~\ref{subsect.logic}.
\end{enumerate*}
 
\subsection{(Unique/affine) existence in three-valued logic}
\label{subsect.existence}

Distributed algorithms often require an implementation to make choices, e.g. `respond only to the first message you receive'.
In axioms, we model this using unique existence and affine existence:
\begin{defn}
\label{defn.unique.exist}
Suppose $\Val$ is a set and $f:\Val\to\THREE$ is a function.
\begin{enumerate*}
\item\label{item.3.equality}
Write ${\teq}:\Val\to\Val\to\THREE$ for \deffont{$\THREE$-equality} which is such that $v\teq v' = \tvT$ if $v=v'$ and $v\teq v' = \tvF$ if $v\neq v'$.\footnote{I hope typography will distinguish between $\teq$ (the function that returns a truth-value in $\THREE$) and $=$ (equality in the discourse of this paper).  
} 
\item
Define \deffont{existence} $\texi f$,\ \deffont{affine existence} $\texiaffine f$ (there exist zero or one) and \deffont{unique existence} $\texiunique  f$ (there exists precisely one)\footnote{The standard symbol for unique existence is $\exists!$.  I am not aware of a standard symbol for affine existence.
\emph{Regular expressions} use $?$ to denote zero or one matches, so I would like to write affine existence as $\exists?$.  Unfortunately we are already using $!$ and $?$ to name the forwards and backward rules, and I found it visually confusing to also place $!$ and $?$ in the logical syntax.  So we write $\existunique$ and $\existaffine$ instead.} %
as per Figure~\ref{fig.ex}.
\end{enumerate*}
\end{defn}

\begin{figure}
$$
\textstyle
\begin{array}{r@{\ }l}
\texi f =& \bigvee_{v\in\Val} f(v)
\\
\texiaffine f =& \bigwedge_{v,v'{\in}\Val}(f(v)\tand f(v')) \tnotor v \teq v'
\\
\texiunique f =&(\texi f) \tand (\texiaffine f) 
\end{array}
$$
Above: $\Val$ is a set and $f:\Val\to\THREE$ is a function; $=$ on the left is definitional equality defining the LHS to mean the RHS (this $=$ is \emph{not} an operation on $\THREE$); $\tnotor$ in the clause for $\texiaffine f$ is the weak implication from Figure~\ref{fig.3}, and $\bigvee$ and $\bigwedge$ denote join (least upper bound) and meet (greatest lower bound) in $\THREE$-the-lattice $\tvF<\tvB<\tvT$ from Definition~\ref{defn.THREE} (so $\tnotor$, $\bigvee$, $\bigwedge$ \emph{are} operations on $\THREE$); and $\teq:\Val\to\Val\to\THREE$ in the clause for $\texiaffine f$ is the $\THREE$-equality from Definition~\ref{defn.unique.exist}(\ref{item.3.equality}).
\caption{Existence, affine existence, and unique existence (Definition~\ref{defn.unique.exist})}
\label{fig.ex}
\end{figure}

\begin{rmrk}
\label{rmrk.unpack.texiunique}
Note that $\texi$, $\texiunique$, and $\texiaffine$ return truth-values in $\THREE$, whereas the usual Boolean quantifiers -- write them $\exists$, $\exists_{1}$, and $\exists_{01}$ -- return truth-values in $\{\bot,\top\}$.
So even if `$\texiaffine$' looks like `$\exists_{01}$', these are different creatures.

Therefore, continuing the notation of Definition~\ref{defn.unique.exist}, we spell out cases to check how the Definition works:
\begin{enumerate*}
\item
If $\Exists{v\neq v'\in\Val}f(v)= \tvT = f(v')$, then 
$\texiunique f =\tvF=\texiaffine f$.
\item
If $\exists_{1}v\in\Val.f(v)=\tvT$, and $\Forall{v'\in\Val}f(v')\in\{\tvT,\tvF\}$, then 
$\texiunique f=\tvT=\texiaffine f$.
\item
If $\exists_{1}v\in\Val.f(v)=\tvT$, and $\Exists{v'\in\Val}f(v')=\tvB$, then 
$\texiunique f=\tvB=\texiaffine f$.
\item\label{item.explain.existsunique.4}
If $\neg\Exists{v\in\Val}f(v)=\tvT$ and $\existunique v'\in\Val.f(v')=\tvB$, then 
$\texiunique f=\tvB$ and $\texiaffine f=\tvT$. 
\item\label{item.two.b}
If $\neg\Exists{v\in\Val}f(v)=\tvT$ and $\exists v,v'\in\Val.v\neq v'\land f(v)=\tvB=f(v')$, then 
$\texiunique f=\tvB=\texiaffine f$.

\noindent
\emph{If $\Forall{v\in\Val}f(v)=\tvB$, then $\texiunique f=\tvB$ -- note that this is \emph{not} $\tvF$.
This illustrates that in our three-valued setting, `unique existence' is not the same as `unique non-false value'.}
\item
If $\Forall{v\in\Val}f(v)=\tvF$, then
$\texiunique f=\tvF$ and $\texiaffine f=\tvT$.
\end{enumerate*}
\end{rmrk}

Our proofs will use unique/affine existence via Proposition~\ref{prop.unique.affine.existence}: 
\begin{prop}
\label{prop.unique.affine.existence}
Suppose $V$ is a set and $f:V\to\THREE$ is a function and $v,v'\in V$.
Then:
\begin{enumerate*}
\item\label{item.unique.affine.existence.01}
The following are equivalent:
\begin{enumerate*}
\item\label{item.unique.affine.existence.01.a}
$\ment\texiaffine f$.
\item
$\texiaffine f\in\{\tvT,\tvB\}$.
\item
$\existaffine v{\in} V.(\ment\modT f(v))$.
\item
$\existaffine v{\in} V.f(v)=\tvT$.
\item\label{item.unique.affine.existence.01.e}
$\Forall{v,v'{\in}V}f(v)=\tvT=f(v') \limp v=v'$.
\end{enumerate*}
\item\label{item.unique.affine.existence.11}
The following are equivalent:
\begin{enumerate*}
\item
$\ment\texiunique f$.
\item
$\exists v{\in} V.(\ment f(v))\ \land \ {\ment\texiaffine f}$.
\end{enumerate*}
\item\label{item.unique.affine.existence.01.T}
The following are equivalent:
\begin{enumerate*}
\item
$\ment\modT\texiaffine f$.
\item
$\existaffine v{\in} V.f(v)\in\{\tvT,\tvB\}$.
\end{enumerate*}
\item\label{item.unique.affine.existence.1}
The following are equivalent:
\begin{enumerate*}
\item
$\ment\modT\texiunique f$.
\item
$\existunique v{\in} V.f(v)=\tvT$ and $\forall v{\in}V.f(v)\neq\tvB$.
\end{enumerate*}
\item\label{item.unique.affine.existence.01.implies}
If $\ment\texiaffine f$ or $\ment\texiunique f$, and $\ment\modT(f(v)\tand f(v'))$,\ then $v=v'$.
\end{enumerate*}
\end{prop}
\begin{proof}
We unpack Definition~\ref{defn.unique.exist} and check the possibilities in Remark~\ref{rmrk.unpack.texiunique}.
\end{proof}

\subsection{A simple modal logic}
\label{subsect.logic}

\begin{figure}
$$
\hspace{-5em}\begin{array}{l@{\quad}r@{\ }l}
\text{Terms} & t ::=& a \mid v \quad (a\in\ns{VarSymb},\ v\in\Val)
\\
\text{Predicates} & \phi::=& \tbot \mid \tneg\phi \mid \phi\tand\phi \mid t\teq t \mid \everyone\phi \mid \Quorum\phi 
\\
&&\phantom{\tbot} \mid \texi a.\phi \mid \texiaffine a.\phi \mid \modTF\phi \mid \tf P(t) \quad (\tf P\in\predsymb)
\\[2ex]
\text{Sugar} &
\begin{array}[t]{r@{\ }l@{\ \ }r@{\ }l@{\ \ }r@{\ }l}
\someone\phi =& \tneg\everyone\tneg\phi
&
\Contraquorum\phi =& \tneg\Quorum\tneg\phi
&
\phi\tor\phi' =& \tneg(\tneg\phi\tand\tneg\phi')
\\
\modB\phi =& \tneg\modTF\phi
&
\modT\phi =& \phi\tand\modTF\phi
&
\phi\toplus\phi' =& (\phi\tand\tneg\phi')\tor(\tneg\phi\tand\phi') \hspace{-5em}
\\
\phi\tnotor \phi' =& \tneg\phi\tor\phi'
&
\phi\timpc\phi' =& \phi\tnotor\modT\phi'
\\
\tall a.\phi =& \tneg\texi a.\tneg\phi
&
\texiunique a.\phi =& \texiaffine a.\phi\tand \texi a.\phi
\\
\correct{\tf P_1,\dots,\tf P_n} =& 
\rlap{$\tall a.(\modTF\tf P_1(a)\tand\dots\tand\modTF\tf P_n(a))\qquad (\tf P_1,\dots,\tf P_n\in\predsymb)$}
\\
\incorrect{\tf P_1,\dots,\tf P_n} =& 
\rlap{$\tall a.(\modB\tf P_1(a)\tand\dots\tand\modB\tf P_n(a))$} 
\end{array}
\hspace{-24.5em}
\end{array}
$$
\caption{Syntax of a simple modal logic (Definition~\ref{defn.logic})}
\label{fig.syntax}
\end{figure}

\begin{defn}
\label{defn.logic}
We fix an infinite set of \deffont{variable symbols} $a,b\in\ns{VarSymb}$. 
We assume a nonempty set of \deffont{values} $v\in\Val$ and a set of \deffont{predicate symbols} $\tf P\in\predsymb$.
Specifically:
\begin{enumerate*}
\item\label{item.bb.sig}
For our example of Bracha Broadcast in Definition~\ref{defn.ThyBB}, we set $\Val$ to be any nonempty set and we take
$\predsymb =\{\broadcast, \echo, \ready, \deliver\}$.
\item\label{item.ca.sig}
For our example of Crusader Agreement in Definition~\ref{defn.ca.algorithm}, we set $\Val=\{0,\botval,1\}$ and we take 
$\predsymb =\{\tf{input}, \echo_1, \echo_2, \tf{output} \}$.
\end{enumerate*}
We then define \deffont{terms} $t$, \deffont{predicates} $\phi$, and syntactic sugar as in Figure~\ref{fig.syntax}.

Write $\phi[a\ssm v]$ for \deffont{substitution} of $v$ for $a$; this is the result of replacing every unbound $a$ in $\phi$ with $v$.
\end{defn}

\begin{nttn}
\label{nttn.abuse.notation}
We may abuse notation in two ways:
\begin{enumerate*}
\item\label{item.abuse.notation.1}
We may write $v$ for quantified variables, e.g. writing $\texi v.\echo(v)$ instead of $\texi a.\echo(a)$.
We may also write an (implicitly) universally quantified variable as $v$ instead of as $a$; e.g. we do this in Figures~\ref{fig.bb.valid} and~\ref{fig.ca.valid}.

The reader is welcome to mentally rewrite $v$ back to $a$, or to decree that the $v$ in Figures~\ref{fig.bb.valid} and~\ref{fig.ca.valid} is actually a variable symbol.\footnote{This kind of overloading between object-level variable symbols and discourse-level variable symbols is not uncommon in mathematical discourse; e.g. if asked what `$x$' is in $\forall x\in\mathbb N.x=x$ the answer `it's a number' would be considered just as valid as `it's a variable symbol denoting a number'.  For extra pedant-points, note that $a$ and $b$ are arguably not actually variable symbols; they are arguably discourse-level variables ranging over distinct object-level variable symbols.  

As my set-theory teacher once advised me when I asked if $x$ in the course was a set, a variable denoting a set, or a variable denoting a variable denoting a set: don't think too hard about this, for that way lies madness.}
Meaning should be quite clear.
\item\label{item.abuse.notation.2}
We may also elide quantified variables entirely, in the style of higher-order logic.
E.g. we may write $\texi \echo$ for $\texi a.\echo(a)$, or $\texiaffine\someone\deliver$ for $\texiaffine a.\someone\deliver(a)$.
We will only do this where it is unambiguous how to fill in the relevant variable symbol if required.
\end{enumerate*}
\end{nttn}

\begin{figure}
$$
\begin{array}{r@{\ }l@{\quad}r@{\ }l@{\quad}r@{\ }l}
\modelm{\tbot}(p)=&\tvF
&
\modelm{\tneg\phi}(p)=&\tneg\modelm{\phi}(p)
\\
\modelm{\phi\tand\phi'}(p) =& \modelm{\phi}(p)\tand\modelm{\phi'}(p)
\\
\modelm{\everyone\phi}(p) =& \everyone \modelm{\phi}
&
\modelm{\Quorum\phi}(p) =& \Quorum \modelm{\phi}
\\
\modelm{\texi a.\phi}(p) =& \texi \lambda v.\modelm{\phi[a\ssm v]}(p)
&
\modelm{\texiaffine a.\phi}(p) =& \texiaffine \lambda v.\modelm{\phi[a\ssm v]}(p)
\\
\modelm{\modTF\phi}(p) =& \modTF(\modelm{\phi}(p))
\\
\modelm{\tf P(v)}(p) =& \varsigma(\tf P)(p)(v)
\ \ \rlap{$(\tf P\in\predsymb)$} 
\\
\modelm{v\teq v'}(p)=& \begin{cases} \tvT & v=v' \\ \tvF & v\neq v' \end{cases} 
\end{array}
$$ 
\ \\
$$
\begin{array}{l@{\ \ }c@{\ \ }l@{\qquad}l@{\ \ }c@{\ \ }l}
p\mentM\phi &\text{means}&\modelm{\phi}(p)\in\{\tvT,\tvB\}
&
\mentM\phi &\text{means}&\Forall{p{\in}\points}(p\mentM\phi)
\\
\mentM\Phi &\text{means}&\Forall{\phi{\in}\Phi}({\mentM\phi})
&
\Phi\mentM\phi &\text{means}&{\mentM\Phi} \limp {\mentM\phi}
\end{array}
$$
Above: $\phi$ is a closed predicate (no free variables; it may contain values).
$\texi$ and $\texiaffine$ to the left of $=$ are syntax from Figure~\ref{fig.syntax}, and to the right of $=$ they are functions as per Figure~\ref{fig.ex}.
Similarly $\everyone$ and $\quorum$ to the left of $=$ are syntax, and to the right of $=$ they are functions as per Definition~\ref{defn.semitopology.logic}.
The $\teq$ in $v\teq v$ on the left is the symbol in our predicate syntax from Figure~\ref{fig.syntax}; the other `equals' symbols are normal mathematical equality. 
\caption{Denotation for a simple modal logic (Definition~\ref{defn.sem})}
\label{fig.sem}
\end{figure}

\begin{defn}
\label{defn.sem}
\leavevmode
\begin{enumerate*}
\item
Given a semitopology $(\points,\opens)$, an \deffont{interpretation} $\varsigma:\predsymb\to\points\to\Val\to\THREE$ is an assignment to each $\tf P\in\predsymb$ of a function $\varsigma(\tf P):\points\to\Val\to\THREE$.
\item
A \deffont{model} $\mu = (\Val,(\points,\opens),\varsigma)$ consists of a set of values, a semitopology, and an interpretation.\footnote{In practice we might wish to include $\predsymb$ explicitly in $\mu$, but note that this is mathematically redundant, in the sense that $\predsymb$ can be recovered as the domain of $\varsigma$.}
\item
Given a model $\mu$, and a closed predicate $\phi$ (no free variable symbols; it may contain values), and $p\in\points$, we define notions of 
\begin{itemize*}
\item
\deffont{denotation} $\modelm{\phi}:\points\to\THREE$ and 
\item
\deffont{validity} $p\mentM\phi$ and $\mentM\phi$ 
\end{itemize*}
as per Figure~\ref{fig.sem}.
\item
An \deffont{axiom} is a closed predicate $\phi$.\footnote{The terminology reflects intent: if we point at a closed predicate $\phi$ and say `This $\phi$ is an axiom!', it carries a connotation that we may intend to restrict to models in which $\phi$ is \emph{valid}.} 
Call an axiom $\phi$ \deffont{valid} in $\mu$ when $\mentM\phi$.
Call a set of axioms $\theory{Thy}$ a \deffont{theory}.
\item
Write $\mentM\theory{Thy}$ when $\Forall{\phi{\in}\theory{Thy}}(\mentM\phi)$ (i.e. the axioms of \theory{Thy} are all valid in $\mu$), in which case we call $\mu$ a \deffont{model of} \theory{Thy}.
\item
Write $\theory{Thy}\mentM\phi$ when $\mentM\theory{Thy}$ implies $\mentM\phi$ (i.e. if $\mu$ is a model of \theory{Thy}, then $\phi$ is valid in $\mu$).
\end{enumerate*}
\end{defn}

\begin{nttn}
In what follows, we may use the fact that $p\mentM \someone\phi$ is equivalent to $\mentM\someone\phi$ without comment, and similarly for the other modalities $\everyone$, $\Quorum$, and $\Contraquorum$, and for compound predicates like $\someone\phi\tnotor\everyone\phi$. 
This pattern will be quite common in the proofs that follow, because many of our axioms have the form $\phi\tnotor\Quorum\phi'$, so when we use \weakmodusponens we may assume $p\mentM\modT\phi$ and deduce $\mentM\Quorum\phi'$ (not $p\mentM\Quorum\phi'$); the $p$ is still there in some sense, but it no longer matters.
\end{nttn}

\begin{rmrk}
We should check that the choice of sugared symbols in Figure~\ref{fig.syntax} matches up with what we would expect from the truth-tables in Figure~\ref{fig.3} and the denotations in Figure~\ref{fig.sem}.
 
Thus (for example) in Figure~\ref{fig.sem} we fix the denotations $\model{\tneg\phi}$ and $\model{\phi\tand\phi'}$, and in Figure~\ref{fig.syntax} we define 
$\phi\tor\phi'$ as sugar for $\tneg(\tneg\phi\tand\tneg\phi')$ and 
$\phi\toplus\phi'$ as sugar for $(\phi\tand\tneg\phi')\tor(\tneg\phi\tand\phi')$ and
$\phi\tnotor\phi'$ as sugar for $\tneg\phi\tor\phi'$ and 
$\modT\phi'$ as sugar for $\phi'\tand\modTF\phi'$ and 
$\phi\timpc\phi'$ as sugar for $\phi\tnotor\modT\phi'$.

So if we unpack $\model{\text{-}}$ on $\tor$, $\toplus$, $\tnotor$, $\modT$, and $\timpc$ as \emph{sugar}, do we arrive at the truth-tables for $\tor$, $\toplus$, $\tnotor$, $\modT$, and $\timpc$ respectively in Figure~\ref{fig.3}?

We do.
Checking this is completely straightforward. 
\end{rmrk}

\section{Declarative Bracha Broadcast}
\label{sect.bb}

\subsection{Definition of the theory}

\begin{rmrk}
\label{rmrk.high-level.bb}
Bracha Broadcast is a classic broadcast algorithm~\cite{bracha:asybap}.
An English description of the algorithm is given in Figure~\ref{fig.bb.english}.
It is designed to allow a group of participants to agree on at most one value, in the presence of some hostile participants. 
\end{rmrk}

\begin{figure}[t]
\noindent\textbf{Protocol description.}
\begin{enumerate*}
\item
A designated \emph{sender} participant sends messages to all participants, \deffont{broadcasting} a value $v$.
Message-passing is reliable, so messages always arrive.
\item
If a participant receives a broadcast $v$ message from the sender -- it is assumed that every participant knows who the sender is, and other participants cannot forge the sender's signature -- it sends an \deffont{echo} $v$ to all other participants.

Each participant will only echo \emph{once}, so if it receives two broadcast messages with different values then it will only echo one of them. 
\item\label{item.bb.contraquorum}
If a participant receives a quorum of echo messages for a value $v$, or a contraquorum of ready messages for a value $v$, then it sends messages to all participants declaring itself \deffont{ready} with $v$.
\item
If a participant receives a quorum of ready messages for a value $v$, then the participant \deffont{delivers} $v$.
\end{enumerate*}
\noindent\textbf{Failure assumptions.}
Our failure assumption will be that a quorum of participants follows the algorithm, and a coquorum (the complement of a quorum, i.e. a nontotal closed set) of participants may not follow the algorithm (more on this terminology in Remark~\ref{rmrk.contraquorum}).
For instance, a hostile participant may broadcast or echo multiple values $v$. 
\caption{Informal (English) description of Bracha Broadcast (Remark~\ref{rmrk.high-level.bb})}
\label{fig.bb.english}
\end{figure}

\begin{figure}[t]
$$
\begin{array}[t]{l@{\quad}r@{\ }l}
\figunderline{Backward rules}
\figskip
\rulefont{BrDeliver?}&
&\deliver(a) \tnotor \QuorumBox\ready(a)
\figskip
\rulefont{BrReady?}&
&\ready(a) \tnotor \QuorumBox\echo(a)
\figskip
\rulefont{BrEcho?}&
&\echo(a) \tnotor \someone\broadcast(a) 
\figskip
\figskip
\figunderline{Other rules}
\figskip
\rulefont{BrEcho01}&
&\texiaffine \echo
\figskip
\rulefont{BrBroadcast1}&
&\texiunique \someone\broadcast
\end{array}
\qquad
\begin{array}[t]{l@{\quad}r@{\ }l}
\figunderline{Forward rules}
\figskip
\rulefont{BrDeliver!}&
&\QuorumBox\ready(a) \tnotor \deliver(a) 
\figskip
\rulefont{BrReady!}&
&\QuorumBox\echo(a) \tnotor \ready(a)
\figskip
\rulefont{BrEcho!}&
&\someone\broadcast(a) \tnotor \texi \echo
\figskip
\rulefont{BrReady!!}&
&\CoquorumDiamond\ready(a) \tnotor \ready(a)
\figskip
\figskip
\rulefont{BrCorrect}&
&\QuorumBox\correct{\ready}\tand\QuorumBox\correct{\echo}
\figskip
\rulefont{BrCorrect'}&
&\correct{\tf P}\tor\incorrect{\tf P}
\ (\tf P{\in}\{\ready,\echo\})
\figskip
\rulefont{BrCorrect''}&
&\everyone\correct{\broadcast}\tor\everyone\incorrect{\broadcast}
\end{array}
$$
Free $a$ above are assumed universally quantified (i.e. there is an invisible $\tall a$ at the start of any axiom with a free $a$).
The axioms above assume predicate symbols $\predsymb =\{\broadcast, \echo, \ready, \deliver\}$, as per Definition~\ref{defn.logic}.
We use Notation~\ref{nttn.abuse.notation}(\ref{item.abuse.notation.2}) to elide quantified variable symbols in \rulefont{BrEcho01} and \rulefont{BrBroadcast1}. 
\caption{$\ThyBB$: axioms of Declarative Bracha Broadcast (Definition~\ref{defn.ThyBB})}
\label{fig.bb}
\end{figure}

\begin{figure}[t]
$$
\begin{array}[t]{l@{\quad}r@{\ }l@{\qquad}l}
\rulefont{BrValidity}&&
\ThyBB\mentM\someone\broadcast(v)\tnotor\everyone\deliver(v)
&\text{(Proposition~\ref{prop.bb.validity})}
\\
\rulefont{BrNoDup}&&
\ThyBB\mentM\texiaffine\deliver
&\text{(Proposition~\ref{prop.no.duplication})}
\\
\rulefont{BrIntegrity}&&
\ThyBB\mentM\deliver(v) \tnotor \someoneAll\broadcast(v)
&\text{(Proposition~\ref{prop.integrity})}
\\
\rulefont{BrConsistency}&&
\ThyBB\mentM\texiaffine\someone\deliver %
&\text{(Proposition~\ref{prop.consistency.bb})}
\\
\rulefont{BrTotality}&&
\ThyBB\mentM\someoneAll\deliver(v) \tnotor \everyoneAll\deliver(v)
&\text{(Proposition~\ref{prop.bb.totality})}
\end{array}
$$
Above, we use Notation~\ref{nttn.abuse.notation}(\ref{item.abuse.notation.2}) to elide quantified variable symbols in \rulefont{BrNoDup} and \rulefont{BrConsistency}. 
\caption{Correctness properties for $\ThyBB$ (Remark~\ref{rmrk.bb.goals})}
\label{fig.bb.valid}
\end{figure}

\begin{defn}
\label{defn.ThyBB}
Define \deffont{Declarative Bracha Broadcast} to be the theory (set of axioms) $\ThyBB$ in Figure~\ref{fig.bb}.
Axioms are taken to be universally quantified over any free variables (there is only one: $a$).
As per Definition~\ref{defn.logic}(\ref{item.bb.sig}), we set $\Val$ to be any nonempty set and $\predsymb =\{\broadcast, \echo, \ready, \deliver\}$.
\end{defn}

\begin{rmrk}
\label{rmrk.model}
A fundamental difference between a logical theory and an algorithm is: 
\begin{itemize*}
\item
With a theory, we are presented with a model and we ask ``Are the axioms valid in this model?''.
\item
With an algorithm, we start from nothing and we construct a model by following the algorithm.
\end{itemize*}
So in this paper, we assume a model is just presented to us, and our business is to check whether it validates axioms.
All of the complexity inherent in studying an implementation -- i.e. a machine to compute such a model -- is elided by design.
\end{rmrk}

\begin{rmrk}
As for our voting example from Figure~\ref{fig.simple}, the axioms in Figure~\ref{fig.bb} are of three kinds:
\begin{enumerate}
\item
\emph{Backward rules} have names ending with $?$ and intuitively represent reasoning of the form ``if \emph{this} happens, then \emph{that} must have happened''.

The reader might find it helpful to think of these as \emph{safety} or \emph{correctness rules}, but `backward rule' is more neutral and descriptive.
Correctness is a theorem \emph{about} a protocol.
A protocol might fail to be correct and still be described using well-formed backward rules; they just specify a protocol for which the desired correctness property does not hold.
\item
\emph{Forward rules} have names ending with $!$ and intuitively represent forward reasoning of the form ``if \emph{this} happens, then \emph{that} must happen''. 

The reader might find it helpful to think of these as \emph{liveness rules}, but `forward rule' is more neutral and descriptive.
A protocol can fail to be live, yet still be described using forward rules.
\item
\emph{Other rules} reflect particular conditions of the algorithms, e.g. having to do with failure assumptions.
\end{enumerate}
\end{rmrk}

\begin{rmrk}[Discussion of the axioms]
We consider the axioms in turn.
Each is implicitly universally quantified over $a$ if required; so the $a$ in \rulefont{BrDeliver?} represents `any value $v\in\Val$'.
The truth-values $\tvT$ and $\tvF$ represent honest behaviour, and $\tvB$ represents dishonest behaviour. 
Each axiom is evaluated at an arbitrary $p\in\points$: %
\begin{enumerate*}
\item
\rulefont{BrDeliver?} says: if $\deliver(v)$ is $\tvT$ at $p$ then for a quorum of participants $\ready(v)$ is $\tvT$ or $\tvB$.

This reflects that in the algorithm as outlined in Remark~\ref{rmrk.high-level.bb}, if $p$ honestly delivers $v$ then a quorum of participants must have sent ready messages for $v$, though some of those messages may have been sent by dishonest participants.
\item
\rulefont{BrReady?} says: if $\ready(v)$ is $\tvT$ at $p$ then for a quorum of participants $\echo(v)$ is $\tvT$ or $\tvB$.
\item
\rulefont{BrEcho?} says: if $\echo(v)$ is $\tvT$, then somebody $\broadcast(v)$ with $\tvT$ or $\tvB$.
\item
\rulefont{BrDeliver!} says: if $\ready(v)$ is $\tvT$ for a quorum of participants, then $\deliver(v)$ is $\tvT$ or $\tvB$. 

This reflects that in the algorithm, if a quorum of honest participants send ready-$v$ messages and if $p$ is honest, then $p$ will send a deliver-$v$ message. 
\item
\rulefont{BrReady!} says: if $\echo(v)$ is $\tvT$ for a quorum of participants, then $\ready(v)$ is $\tvT$ or $\tvB$. 
\item
\rulefont{BrEcho!} says: if $\broadcast$ is $\tvT$ for some participant and some value, then $\echo$ is $\tvT$ or $\tvB$ for some value, though not necessarily the same one. 

This reflects that in the algorithm, if somebody honestly broadcasts some value and if $p$ is honest, then $p$ will echo some value.
The fact that this will be the \emph{same} value (if both parties are honest) is a Lemma that follows using \rulefont{BrEcho?}; see Lemma~\ref{lemm.bb.forward}(\ref{item.bb.forward.echo'}).
\item
\rulefont{BrReady!!} says: if $\ready(v)$ is $\tvT$ for a \emph{contraquorum} of participants, then $\ready(v)$ is $\tvT$ or $\tvB$.
This reflects the disjunct `or a contraquorum of ready messages' in Remark~\ref{rmrk.high-level.bb}(\ref{item.bb.contraquorum}).
\item
\rulefont{BrEcho01} says that $\echo(v)$ is $\tvT$ for at most one value $v$.
It asserts nothing about byzantine behaviour: $\echo(v)$ can be $\tvB$ for as many values as we like, reflecting the fact that dishonest participants may deviate from the protocol.
\item
\rulefont{BrBroadcast1} says that there is \emph{precisely} one value $v$ such that some participant $\tvT$-broadcasts $v$ (so a model in which two participants $\tvT$-broadcast distinct values would \emph{not} satisfy this condition, whereas a model in which all participants $\tvT$-broadcast the same value and $\tvF$-broadcast all other values, would) \emph{or} there is \emph{at least} one value $v$ such that some participant $\tvB$-broadcasts $v$.
This is what $\texiunique$ means in a three-value setting, as discussed in Remark~\ref{rmrk.unpack.texiunique} and Proposition~\ref{prop.unique.affine.existence}.
\item
\rulefont{BrCorrect} asserts that there are quorums of honest participants for $\ready$ and $\echo$.
The literature tends to assume something slightly stronger, that there is one quorum of honest participants for all predicates.
We will only need this weaker assumption for our proofs.
\item
\rulefont{BrCorrect'} asserts that a participant $p$ is either entirely honest for $\ready$ and $\echo$, in that it returns $\tvT$ or $\tvF$ for every value, or it is entirely dishonest, in that it returns $\tvB$ for every value.
This accurately reflects an informal assumption that (for a given predicate) a participant is either honest or dishonest -- dishonesty is not attached to a particular message, but to the \emph{participant} who sent it. 

\rulefont{BrCorrect''} asserts that either for all participants the $\broadcast$ predicate returns $\tvT$ or $\tvF$ on every $v$, or for all participants it returns $\tvB$.
This reflects that we do not designate a unique sender:
\end{enumerate*}
\end{rmrk}

\begin{rmrk}
\label{rmrk.unique.sender}
Implementations of Bracha Broadcast assume a unique \emph{sender} participant whose identity is known to all participants and whose identity cannot be faked.
The sender's task is to pick a unique value $v$ and broadcast $v$ (and, if the sender is honest, only $v$) to all participants.

In our logic, we model this with a straightforward axiom $\texiunique a.\someone\broadcast(a)$.
Our logic is three-valued, and as per Proposition~\ref{prop.unique.affine.existence}(\ref{item.unique.affine.existence.01.implies}) this means that 
\begin{itemize*}
\item
there is at most one $v$ such that there exists some participant that $\tvT$-sends $v$, but 
\item
there may be many values $v_1,\dots$ such that for each $v_i$ there exists some participant that $\tvB$-sends $v_i$.
\end{itemize*}
We return to this in Remark~\ref{rmrk.discuss.send}, after we have a few more proofs.
\end{rmrk}

\begin{rmrk}
Formally, $\mentM\ThyBB$ means that $\mu$ is a model and the axioms of $\ThyBB$ are valid of $\mu$.
Intuitively, $\mentM\ThyBB$ is intended to mean ``$\mu$ represents a run of Bracha Broadcast'', such that: 
\begin{itemize*}
\item
$p\mentM\broadcast(v)$ means ``$p$ broadcast $v$ at some point in the run''. 
\item
$p\mentM\echo(v)$ means ``$p$ echoed $v$ at some point in the run''. 
\item
$p\mentM\ready(v)$ means ``$p$ declared itself ready for $v$ at some point in the run''. 
\item
$p\mentM\deliver(v)$ means ``$p$ delivered $v$ at some point in the run''. 
\end{itemize*}
Note however that these are just intuitions.
The model is just a mathematical structure; $\broadcast(v)$, $\echo(v)$, $\ready(v)$, and $\deliver(v)$ are not messages, they are predicates.
As per Remark~\ref{rmrk.model} we do not assume anything about where $\mu$ came from and in particular, there is no \emph{a priori} guarantee that $\mu$ was \emph{actually generated} by a run of an actual implementation of the Bracha Broadcast algorithm.
$\mu$ is just a structure of which we assume that $\mentM\ThyBB$.
This abstraction is a feature, not a bug; $\mu$ is intended to be abstract.

Proving (or to use a more suggestive term: \emph{formally verifying}) that a \emph{particular} concrete implementation of Bracha Broadcast \emph{actually does} lead to models that satisfy $\ThyBB$, is the familiar problem of checking that an implementation satisfies its specification.
Consider for example proving that a \emph{particular} implementation of a smart contract \emph{actually does} satisfy a token specification~\cite{gabbay_et_al:OASIcs.FMBC.2021.2}, or that a \emph{particular} implementation of addition \emph{actually does} add numbers.%
\footnote{Note that this can be more subtle than it sounds, even on an apparently innocuous example like addition on numbers.  For example: is 65535+1 equal to 65536, or 0, or an overflow error?  
Perhaps 65535 is not a number at all (it is not if our notion of number is 8-bit integers). 
So the answer depends on what kind of `addition' we mean, on what kind of `number'.
Behind the deceptively innocuous English phrase `a \emph{particular} implementation of addition' can lie many possible formal meanings; signed integer of various bitlengths, unsigned integers, floating point datatypes, arbitrary precision integers, finite field elements, real numbers, complex numbers, quaternions, \dots.  

At this point, we see that careful abstraction and logical specification are not an indulgence but a necessary default, and two specifications being different does not necessarily mean that one of them must be wrong; they may be capturing different formal entities that mere English may gloss over and/or struggle to faithfully express.}
So $\ThyBB$ is a way to declaratively specify what it is about an implementation that makes it be an implementation \emph{of} Bracha Broadcast, and to the extent that we may agree that $\ThyBB$ does indeed capture this intuition in formal axiomatic form, $\ThyBB$ can claim to \emph{be} `Declarative Bracha Broadcast'.
\end{rmrk}

\subsection{Correctness properties for Bracha Broadcast} 
\label{subsect.bb.correctness.properties}

\subsubsection{Statement of the correctness properties}

\begin{rmrk}
\label{rmrk.fix.model.bb}
For the rest of this section, we fix a model (Definition~\ref{defn.sem}) 
$$
\mu=(\Val,(\points,\opens),\varsigma)
$$ 
such that the underlying semitopology $(\points,\opens)$ is 3-twined (Definition~\ref{defn.3twined}), so that we have Theorem~\ref{thrm.3twined.logic}.
\end{rmrk}

\begin{rmrk}
\label{rmrk.bb.goals}
As promised (see end of Remark~\ref{rmrk.high-level.bb}), we will now use \ThyBB in Figure~\ref{fig.bb} to derive predicates corresponding to the following standard correctness properties for Bracha Broadcast, per Figure~\ref{fig.bb.valid}.
We consider them in turn:
\begin{enumerate*}
\item
\textbf{Validity:} \emph{If a correct $p$ broadcasts a value $v$, then every correct $p'$ delivers $v$.} 
\\
This becomes $\ThyBB\mentM\someone\broadcast(v)\tnotor\everyone\deliver(v)$. 
See Proposition~\ref{prop.bb.validity}.
\item
\textbf{No duplication:} \emph{Every correct $p$ delivers at most one value.}
\\
This becomes $\ThyBB\mentM\texiaffine\deliver$.
See Proposition~\ref{prop.no.duplication}.
\item
\textbf{Integrity:}\quad 
\emph{If some $p$ delivers a message $m$ with correct sender $p'$, then $m$ was broadcast by $p'$.}
\\
This becomes $\ThyBB\mentM\deliver(v) \tnotor \someoneAll\broadcast(v)$.
See Proposition~\ref{prop.integrity}.
\item
\textbf{Consistency:} \emph{If correct $p$ delivers $v$ and another correct $p'$ delivers $v'$ then $v{=}v'$.}
\\
This becomes $\ThyBB\mentM\texiaffine\someone\deliver$. 
See Proposition~\ref{prop.consistency.bb}.
\item
\textbf{Totality:} \emph{If correct $p$ delivers $v$, then every correct $p'$ delivers $v$.}
\\
This becomes $\ThyBB\mentM\someoneAll\deliver(v) \tnotor \everyoneAll\deliver(v)$.
See Proposition~\ref{prop.bb.totality}.
\end{enumerate*}
The English statements above are (in order) adapted from properties BCB1, BCB2, BCB3, and BCB4 from Module~3.11, and BRB5 from Module~3.12 of~\cite{cachinbook}.
\end{rmrk}

\begin{rmrk}
We can make a few observations about the correctness properties above:
\begin{enumerate*}
\item
The logical statements are concise, precise, and (arguably) clear.
\item
The rendering of Consistency $\texiaffine\someone\deliver$ in \rulefont{BrConsistency} subsumes that for No Duplication $\texiaffine\deliver$ in \rulefont{BrNoDup}.

This is less visible in the English because Consistency talks about a `correct $p$' and `another correct $p'$', which might be taken to indicate that $p\neq p'$.
In our logic, 
it is easier to unify Consistency and No Duplication into a single predicate and proof.
\item
The English statement of Integrity was hard to parse, at least for me.
Contrast with the predicate $\deliver(v) \tnotor \someoneAll\broadcast(v)$, which is clear and precise.
\end{enumerate*}
\end{rmrk}

\subsubsection{Validity: ``if a correct $p$ broadcasts $v$, then every correct $p'$ delivers $v$''}

\begin{lemm}
\label{lemm.send}
Suppose $\mentM\ThyBB$.
Then precisely one of the following holds:
\begin{enumerate*}
\item\label{item.send.1}
$(\forall p\in\points.p\mentM\correct{\broadcast})\,\tand\,\existunique v\in\Val.\forall p\in\points.(p\mentM \modT\someone\broadcast(v))$,\ or
\item\label{item.send.2}
$\forall v\in\Val,p\in\points.(p\mentM \modB\broadcast(v))$.
\end{enumerate*}
\end{lemm}
\begin{proof}
\rulefont{BrCorrect''} asserts that \emph{either} for all participants the $\broadcast$ predicate returns $\tvT$ or $\tvF$ on every $v$, \emph{or} for all participants it returns $\tvB$.
In the latter case we have the latter (second) disjunct above, and in the former case
using \rulefont{BrBroadcast1} we obtain the former (first) disjunct.

Furthermore, these two disjuncts are mutually exclusive by non-emptiness of $\points$, as per Definition~\ref{defn.semitopology}(\ref{item.semitopology.points}) (cf. Remark~\ref{rmrk.nonempty.P}).
\end{proof}

\begin{rmrk}
\label{rmrk.discuss.send}
Lemma~\ref{lemm.send} requires some discussion.
Our description of Bracha Broadcast in Remark~\ref{rmrk.high-level.bb} assumes a designated sender with an unforgeable signature, who might be byzantine.
This is an accurate description of the implementation, but logically we do not care that there is one sender.
There could be \emph{two} senders, so long as (if honest) they broadcast identical values.

Implementationally we guarantee this by designating a unique sender participant.
In the axiomatisation, it is convenient (and more general) to be more abstract: we axiomatise a broadcast predicate that is either somewhere $\tvT$ for \emph{precisely} one value (corresponding to a value broadcast by an honest sender), or it is $\tvB$ for all values and all participants (corresponding to arbitrary hostile behaviour of a dishonest sender).\footnote{A longer stronger axiomatisation that is closer to the implementation could be written but would be more complex and not buy us much, since we can prove our correctness results from the weaker and simpler formulation.} 
\end{rmrk}

A technical fact will be helpful for Lemma~\ref{lemm.bb.forward}:
\begin{lemm}
\label{lemm.key.technical.lemma}
$\existunique v\in\Val.\forall p\in\points.(p\mentM \modT\someone\broadcast(v))$
implies 
$\forall p\in\points.\existunique v\in\Val.(p\mentM \modT\someone\broadcast(v))$.
\end{lemm}
\begin{proof}
It is simplest to consider the contrapositive: suppose the negation of the right-hand side.
\begin{itemize*}
\item
Suppose there exists $p\in\points$ and two distinct values $v\neq v'\in\Val$ such that $p\mentM\modT\someone\broadcast(v)$ and $p\mentM\modT\someone\broadcast(v')$.
But then it follows that $p'\mentM\modT\someone\broadcast(v)$ for \emph{every} $p'\in\points$, and also $p'\mentM\modT\someone\broadcast(v')$ for \emph{every} $p'\in\points$, which falsifies the left-hand side.
\item
Suppose there exists $p\in\points$ such that for no $v\in\Val$ is it the case that $p\mentM\modT\someone\broadcast(v)$. 
But then it follows that for \emph{every} $p'\in\points$ it is not the case that $p'\mentM\modT\someone\broadcast(v)$, and the result follows.
\qedhere\end{itemize*}
\end{proof}

\begin{lemm}
\label{lemm.bb.forward}
Suppose $\mentM\ThyBB$ and $v\in\Val$.
Then:
\begin{enumerate*}
\item\label{item.bb.forward.echo'}
$\mentM \someone\broadcast(v)\tnotor\echo(v)$. 
\item\label{item.bb.forward.echo}
$\mentM \someone\broadcast(v)\tnotor\everyoneAll\echo(v)$. 
\item\label{item.bb.forward.ready}
$\mentM \QuorumBox\echo(v)\tnotor\everyoneAll\ready(v)$. 
\item\label{item.bb.forward.deliver}
$\mentM \QuorumBox\ready(v)\tnotor\everyoneAll\deliver(v)$. 
\end{enumerate*}
\end{lemm}
\begin{proof}
We consider each part in turn, freely using \weakmodusponens and the semantics in Figure~\ref{fig.sem}:
\begin{enumerate}
\item
Consider $p\in\points$ and suppose $p\mentM\modT\someone\broadcast(v)$.
By \weakmodusponens and \rulefont{BrEcho!} $p\mentM\texi a.\echo(a)$.
Thus, there exists $v'\in\Val$ such that $p\mentM\echo(v')$.
There are now two subcases:
\begin{itemize*}
\item
\emph{Suppose $p\mentM\modB\echo(v')$.}\ 
By \rulefont{BrCorrect'} $p\mentM\modB\echo(v)$, and so $p\mentM\echo(v)$.
\item
\emph{Suppose $p\mentM\modTF\echo(v')$.}\ 
By Proposition~\ref{prop.tand.tor}(\ref{item.para.ment})
(since $p\mentM\echo(v')$) $p\mentM\modT\echo(v')$, and so by \rulefont{BrEcho?} $p\mentM\someone\broadcast(v')$. 
Now we assumed above that $p\mentM\modT\someone\broadcast(v)$ so this puts us in case~\ref{item.send.1} of Lemma~\ref{lemm.send}.
For clarity we write that case out in full: 
$$
\forall p\in\points.p\mentM\correct{\broadcast} \ \tand\ \existunique v\in\Val.\forall p\in\points.(p\mentM \modT\someone\broadcast(v)) .
$$
From the left-hand conjunct and $p\mentM\someone\broadcast(v')$ and Proposition~\ref{prop.tand.tor}(\ref{item.para.ment}) we have that $p\mentM\modT\someone\broadcast(v')$. 
But we assumed $p\mentM\modT\someone\broadcast(v)$ and so by the right-hand conjunct and Lemma~\ref{lemm.key.technical.lemma}
we have $v=v'$ as required.
\end{itemize*}
\item
From part~\ref{item.bb.forward.echo'} of this result, noting that the $p\in\points$ we chose in the proof was arbitrary.
\item
Suppose $p\mentM\modT\QuorumBox\echo(v)$.
By \rulefont{BrReady!} $p\mentM\ready(v)$.
This reasoning did not depend on $p$, so $\mentM\everyoneAll\ready(v)$ as required.
\item
As the reasoning for part~\ref{item.bb.forward.ready}, but using \rulefont{BrDeliver!}.
\qedhere\end{enumerate}
\end{proof}

\begin{rmrk}
\label{rmrk.bb.forward.discuss}
For clarity we spell out what Lemma~\ref{lemm.bb.forward} asserts.
Each part has the form $\mentM \phi\tnotor\psi$, meaning that for every $p\in\points$ if $p\mentM\modT\phi$ (i.e. $\modelm{\phi}(p)=\tvT$) then $p\mentM\psi$ (i.e. $\modelm{\psi}(p)\in\{\tvT,\tvB\}$).
So part~\ref{item.bb.forward.echo'} means 
`$p\mentM \modT\someone\broadcast(v)$ implies $p\mentM \echo(v)$'. 

Note it does \emph{not} mean `$p\mentM \someone\broadcast(v)$ implies ${p\mentM \echo(v)}$'.
That would be a different assertion.
\end{rmrk}

\begin{lemm}
\label{lemm.everyone.to.quorum.true}
Suppose $\mentM\ThyBB$ and $v\in\Val$.
Then:
\begin{enumerate*}
\item\label{item.everyone.to.quorum.true.echo}
$\mentM\everyoneAll\echo(v)$ implies $\mentM\modT\QuorumBox\echo(v)$.
\item\label{item.everyone.to.quorum.true.ready}
$\mentM\everyoneAll\ready(v)$ implies $\mentM\modT\QuorumBox\ready(v)$.
\end{enumerate*}
\end{lemm}
\begin{proof}
We consider each part in turn:
\begin{enumerate}
\item
Suppose $\mentM\everyoneAll\echo(v)$.
By \rulefont{BrCorrect} $\mentM\Quorum\correct{\echo}$ -- by Figure~\ref{fig.syntax} this means $\mentM\Quorum\tall v.\modTF\echo(v)$ -- thus $\mentM\Quorum\modTF\echo(v)$ and by Lemma~\ref{lemm.everyone.and.quorum}(\ref{item.everyone.and.quorum.2}) and Proposition~\ref{prop.tand.tor}(\ref{item.para.ment}) 
$\mentM\modT\QuorumBox\echo(v)$. 
\item
As for part~\ref{item.everyone.to.quorum.true.echo}, since by \rulefont{BrCorrect} $\mentM\Quorum\correct{\ready}$.
\qedhere
\end{enumerate}
\end{proof}

\begin{prop}[Validity]
\label{prop.bb.validity}
If $\mentM\ThyBB$ and $v\in\Val$ then
$$
\mentM\someone\broadcast(v)\tnotor\everyone\deliver(v). 
$$
\end{prop}
\begin{proof}
We reason using \weakmodusponens. 
Suppose $\mentM\modT\someone\broadcast(v)$.
By Lemma~\ref{lemm.bb.forward}(\ref{item.bb.forward.echo}) $\mentM\everyoneAll\echo(v)$,
and by Lemma~\ref{lemm.everyone.to.quorum.true}(\ref{item.everyone.to.quorum.true.echo}) $\mentM\modT\QuorumBox\echo(v)$.
By Lemma~\ref{lemm.bb.forward}(\ref{item.bb.forward.ready}) $\mentM\everyoneAll\ready(v)$, and by Lemma~\ref{lemm.everyone.to.quorum.true}(\ref{item.everyone.to.quorum.true.ready})
$\mentM\modT\QuorumBox\ready(v)$. 
By Lemma~\ref{lemm.bb.forward}(\ref{item.bb.forward.deliver}) $\mentM\everyoneAll\deliver(v)$ as required.
\end{proof}

\subsubsection{Consistency \& No duplication: ``if correct $p$ delivers $v$ and correct (possibly equal) $p'$ delivers $v'$ then $v=v'$''}

\begin{lemm}
\label{lemm.BB.box.diamond}
Suppose $\mentM\rulefont{BrCorrect}$ (in particular, it suffices that $\mentM\ThyBB$). 
Then for $v,v'\in\Val$:
\begin{enumerate*}
\item\label{item.BB.box.diamond.box.diamond} 
If $\mentM\QuorumBox\ready(v)$ then $\mentM\modT\Contraquorum\ready(v)$.
\item\label{item.BB.box.diamond.box.box.diamond.echo} 
If $\mentM\QuorumBox\echo(v) \tand \QuorumBox\echo(v')$ then $\mentM\modT\someoneAll(\echo(v)\tand\echo(v'))$.
\item\label{item.BB.box.diamond.box.box.diamond.ready} 
If $\mentM\QuorumBox\ready(v) \tand \QuorumBox\ready(v')$ then $\mentM\modT\someoneAll(\ready(v)\tand\ready(v'))$.
\end{enumerate*}
\end{lemm}
\begin{proof}
We consider each part in turn:
\begin{enumerate}
\item
Suppose $\mentM\QuorumBox\ready(v)$.
From \rulefont{BrCorrect} also $\mentM\Quorum\modTF\ready(v)$.
Using Theorem~\ref{thrm.3twined.logic} $\mentM\Contraquorum\modT\ready(v)$ and so $\mentM\modT\Contraquorum\ready(v)$ as required. 
\item
Suppose $\mentM\QuorumBox\echo(v) \tand \QuorumBox\echo(v')$.
From \rulefont{BrCorrect} we have $\mentM\quorum\correct{\echo}$.
Using Theorem~\ref{thrm.3twined.logic} and Lemma~\ref{lemm.semi.char}(\ref{item.semi.char.TF}) (or just direct from the 3-twined property in Definition~\ref{defn.3twined})
$\mentM\modT\someoneAll(\echo(v)\tand\echo(v'))$ as required.
\item
As for the proof of part~\ref{item.BB.box.diamond.box.box.diamond.echo}, noting that $\mentM\quorum\correct{\ready}$ from \rulefont{BrCorrect}. 
\qedhere\end{enumerate}
\end{proof}

\begin{prop}[Consistency \& No Duplication]
\label{prop.no.duplication}
\label{prop.consistency.bb}
Suppose $\mentM\ThyBB$.
Then
$$
\mentM\texiaffine v.\someone\deliver(v) .
$$
\end{prop}
\begin{proof}
By Proposition~\ref{prop.unique.affine.existence}(\ref{item.unique.affine.existence.01}) it suffices to show for every $v,v'\in\Val$ that
$\mentM\modT\someoneAll\deliver(v)\tand\modT\someoneAll\deliver(v')$ implies $v= v'$.
So suppose $\mentM\modT\someoneAll\deliver(v)$ and $\mentM\modT\someoneAll\deliver(v')$. 
By \rulefont{BrDeliver?} $\mentM\QuorumBox\ready(v)$ and $\mentM\QuorumBox\ready(v')$, so by Lemma~\ref{lemm.BB.box.diamond}(\ref{item.BB.box.diamond.box.box.diamond.ready}) 
$\mentM\modT\someone(\ready(v)\tand\ready(v'))$.
By \rulefont{BrReady?} $\mentM\QuorumBox\echo(v)$ and $\mentM\QuorumBox\echo(v')$, so by Lemma~\ref{lemm.BB.box.diamond}(\ref{item.BB.box.diamond.box.box.diamond.echo})
$\mentM\modT\someone(\echo(v)\tand\echo(v'))$.
It follows using \rulefont{BrEcho01} and Proposition~\ref{prop.unique.affine.existence}(\ref{item.unique.affine.existence.01.implies}) that $v=v'$ as required. 
\end{proof}

\subsubsection{Integrity: ``if some $p$ delivers a message $m$ with correct sender $p'$, then $m$ was broadcast by $p'$''}

\begin{prop}[Integrity]
\label{prop.integrity}
Suppose $\mentM\ThyBB$ and $v\in\Val$.
Then 
$$
\mentM\deliver(v) \tnotor \someoneAll\broadcast(v) .
$$
\end{prop}
\begin{proof}
We reason using \weakmodusponens.
Suppose $p\in\points$ and $p\mentM\modT\deliver(v)$. 
By \rulefont{BrDeliver?} $\mentM\QuorumBox\ready(v)$. 
By Lemma~\ref{lemm.BB.box.diamond}(\ref{item.BB.box.diamond.box.box.diamond.ready}) (taking $v=v'$ in that Lemma) $\mentM\modT\someoneAll\ready(v)$. 
By \rulefont{BrReady?} $\mentM\QuorumBox\echo(v)$, and by Lemma~\ref{lemm.BB.box.diamond}(\ref{item.BB.box.diamond.box.box.diamond.echo}) (taking $v=v'$) $\mentM\modT\someoneAll\echo(v)$. 
Using \rulefont{BrEcho?} $\mentM\someoneAll\broadcast(v)$, as required.
\end{proof}

\subsubsection{Totality: if correct $p$ delivers $v$, then every correct $p'$ delivers $v$}

\begin{prop}[Totality]
\label{prop.bb.totality}
Suppose $\mentM\ThyBB$ and $v\in\Val$.
Then: 
$$
\mentM \someoneAll\deliver(v) \ \tnotor\ \everyoneAll\deliver(v).
$$ 
\end{prop}
\begin{proof}
We reason using \weakmodusponens.
Suppose $\mentM \modT\someoneAll\deliver(v)$.
By \rulefont{BrDeliver?} $\mentM\QuorumBox\ready(v)$, so by Lemma~\ref{lemm.BB.box.diamond}(\ref{item.BB.box.diamond.box.diamond})
$\mentM\modT\CoquorumDiamond\ready(v)$.
Using \rulefont{BrReady!!} (applied at every point) $\mentM\everyoneAll\ready(v)$, so by Lemma~\ref{lemm.everyone.to.quorum.true}(\ref{item.everyone.to.quorum.true.ready}) $\mentM\modT\QuorumBox\ready(v)$.
Using \rulefont{BrDeliver!} (applied at every point) $\mentM\everyoneAll\deliver(v)$. 
\end{proof}

\begin{rmrk}
Let us take a moment to unpack the statement of Totality in Proposition~\ref{prop.bb.totality} and remember what the symbols mean.
As per Figure~\ref{fig.sem} this asserts that $\modelm{\someone\deliver(v) \tnotor \everyone\deliver(v)}(p)$ has truth-value $\tvT$ or $\tvB$ for every $p\in\points$.
The choice of $p$ clearly does not matter here because both parts of the implication are modal.
Unpacking the truth-table of $\tnotor$ from Figure~\ref{fig.3}, it suffices to show that if the left-hand side has truth-value $\tvT$ then the right-hand side has truth-value $\tvT$ or $\tvB$.
 
So suppose $\modelm{\someone\deliver(v)}=\tvT$, indicating that some honest participant delivers $v$.
The proof of Proposition~\ref{prop.bb.totality} then shows that:
every other honest participant delivers $v$;
for any dishonest participants $\deliver(v)$ has truth-value $\tvB$;
and therefore, the truth-value of $\everyone\deliver(v)$ is $\tvT$ or (if dishonest participants are involved) $\tvB$.

The statement and proof of Proposition~\ref{prop.bb.totality} do not reason explicitly on correct or byzantine participants.
They do not need to: this detail is all packaged neatly into how the modalities and implication work.
The proof is routine, following the structure of the axioms and applying lemmas in a natural way. 

It is easy to take this magic for granted: how logic turns reasoning into (almost) routine symbolic manipulation.
But of course, this is why logical specifications can be so helpful.
Declarative logical specifications of the kind we see above are particularly powerful, because of how concise they can be.
\end{rmrk}

\subsection{Further discussion of the axioms}

We take a moment to return to the axioms in Figure~\ref{fig.bb}, to discuss some design alternatives.

Axioms \rulefont{BrDeliver!} and \rulefont{BrReady!} in Figure~\ref{fig.bb} are slightly stronger than absolutely necessary.
In the presence of the other axioms, it suffices to assume $\texi a.\Quorum\ready(a) \tnotor \texi a.\deliver(a)$ and $\texi a.\Quorum\echo(a)\tnotor\texi a.\ready(a)$.
Then it is not hard to use the backward rules to derive the stronger forms of the axioms as presented in Figure~\ref{fig.bb}.

The same is not true of \rulefont{BrReady!!}.
A weaker axiom $\texi a.\Contraquorum\ready(a) \tnotor \texi a.\ready(a)$ would \emph{not} suffice and the stronger form of the axiom cannot be derived from it. 

Conversely, changing \rulefont{BrEcho!} to $\someone\broadcast(a)\tnotor \echo(a)$ would be wrong.
Consider the case of a byzantine participant who (dishonestly; contrary to the protocol) broadcasts $v$ and $v'$.
In our logic this would give $\someone\broadcast(v)$ and $\someone\broadcast(v')$ truth-value $\tvB$.
According to the truth-table for $\tnotor$ in Figure~\ref{fig.3}, by our modified \rulefont{BrEcho!} rule, $\echo(v)$ and $\echo(v')$ would then have to have truth-value $\tvB$ or $\tvT$.
For an honest participant this would mean that both $\echo(v)$ and $\echo(v')$ would have to have truth-value $\tvT$.
This would contradict \rulefont{BrEcho01}, which implies that no honest participant may echo more than one value.

Clause~\ref{item.bb.contraquorum} of Remark~\ref{rmrk.high-level.bb} renders into logic as a forward rule $(\Quorum\echo(a)\tor\Contraquorum\ready(a))\tnotor \ready(a)$.
This rule appears more-or-less verbatim in Figure~\ref{fig.bb} as \rulefont{BrReady!} and \rulefont{BrReady!!}.
Yet in Figure~\ref{fig.bb} we only have the backward rule \rulefont{BrReady?} as $\ready(a)\tnotor\Quorum\echo(a)$. 
Why not a rule \rulefont{BrReady?'} of the form $\ready(a) \tnotor (\Quorum\echo(a)\tor\Contraquorum\ready(a))$?
Algorithmically this is inefficient, because if somebody sends a ready message then by an inductive argument a quorum of echoes must exist. 
In the logic, \rulefont{BrReady?'} would actually be \emph{too weak}; it would permit a model in which every participant does $\ready(v)$, and for each $p$ \rulefont{BrReady?'} would be valid at $p$ because of the contraquorum of all the \emph{other} $p'$ who also do $\ready(v)$.
This is related to the observations in Remark~\ref{rmrk.model}: 
in the algorithmic world we build things, so things \emph{can't} happen unless we say they \emph{do}; in the logical world we verify things, so things \emph{can} happen unless we say they \emph{don't}.

\section{Crusader Agreement}
\label{sect.crusader.agreement}

\subsection{Motivation and presenting the protocol}

\begin{rmrk}
In Section~\ref{sect.bb} we gave a declarative axiomatisation of the Bracha Broadcast protocol and %
we showed how to prove its correctness properties by axiomatic reasoning. 

We will now do the same for a Crusader Agreement algorithm as presented in an online exposition~\cite{abraham:ca}.
This exposition was written for the Decentralized Thoughts website\footnote{%
\href{https://decentralizedthoughts.github.io/}{Decentralized Thoughts} is a well-regarded source of technical commentary on decentralised protocols.} by the authors of a journal paper~\cite{abraham:effasa}, to provide an accessible introduction to their paper.

Agreement is arguably a moderate step up in complexity from Broadcast, because:
\begin{enumerate*}
\item
With Broadcast, we want all correct participants to agree on a value broadcast by a single designated (possibly faulty) broadcasting participant.
\item
With Agreement, we want all correct participants to agree on some value, where all participants start off with their own individual values. 
While still simple, it is slightly more challenging in the sense that correct participants can propose different values whereas in Broadcast they cannot.
\end{enumerate*}
\end{rmrk}

\begin{defn}
\label{defn.ca.algorithm}
\leavevmode
\begin{enumerate*}
\item
Figure~\ref{fig.ca.code} presents the pseudocode from~\cite{abraham:ca}.
\item
Figure~\ref{fig.ca} presents \deffont{Declarative Crusader Agreement} to be the theory (set of axioms) \ThyCA.
Axioms are taken to be universally quantified over any free variables (there is only one: $a$).
As per Definition~\ref{defn.logic}(\ref{item.ca.sig}), we set $\Val=\{0,\botval,1\}$ and $\predsymb =\{\tf{input}, \echo_1, \echo_2, \tf{output} \}$.
\end{enumerate*}
\end{defn}

\begin{rmrk}
The authors of~\cite{abraham:ca} are not completely explicit, but it becomes clear from context, that:
\begin{enumerate*}
\item
Clauses in Figure~\ref{fig.ca.code} are intended to be executed in parallel (not necessarily sequentially).
\item
Each participant starts with their own input value, which may be different from the input value(s) of other participants.
\end{enumerate*}
Intuitively, in this algorithm participants try to agree on a value $0$ or $1$, and if a participant can see no clear choice of value to agree on, then it might output a special `agreement failed' value.
Here we write this value as $\botval$; in~\cite{abraham:ca} they write $\tbot$ (a riff on the domain-theoretic `undefined' element), but that clashes with our use in this paper of $\tbot$ for logical falsity in Figure~\ref{fig.syntax}.
\end{rmrk}

\begin{figure}
\begin{Verbatim}[numbers=left,xleftmargin=10mm]
input: v (0 or 1)
send <echo1, v> to all
if didnt send <echo1, 1-v> and see f+1 <echo1, 1-v>
    send <echo1, 1-v> to all
if didnt send <echo2, *> and see n-f <echo1, w>
    send <echo2, w>
if see n-f <echo2, u> and n-f <echo1, u>
    output(u)
if see n-f <echo1, 0> and n-f <echo1, 1>
    output(0.5)
\end{Verbatim}
\caption{Pseudocode Crusader Agreement algorithm from~\cite{abraham:ca} (Definition~\ref{defn.ca.algorithm})}
\label{fig.ca.code}
\end{figure}

\begin{figure}
$$
\begin{array}[t]{l@{\quad}r@{\ }l}
\figunderline{Backward rules}
\figskip
\rulefont{CaEcho1?}&
&\echo_1(a)\timpc\someone\tf{input}(a)
\\
\rulefont{CaEcho2?}&
&\echo_2(a)\tnotor\quorum\echo_1(a)
\\
\rulefont{CaOutput?}&
&(\tf{output}(0)\tnotor\quorum\echo_2(0))\tand(\tf{output}(1)\tnotor\quorum\echo_2(1))
\\
\rulefont{CaOutput'?}&
&
\tf{output}(\botval)\tnotor(\quorum\echo_1(0)\tand\quorum\echo_1(1))
\figskip
\figskip
\figunderline{Forward rules}
\figskip
\rulefont{CaEcho1!}&
&(\tf{input}(a)\tor\contraquorum\echo_1(a))\tnotor\echo_1(a)
\\
\rulefont{CaEcho2!}&
&(\texi \quorum\echo_1) \tnotor \texi \echo_2
\\
\rulefont{CaOutput!}&
&\quorum\echo_2(a)\tnotor \tf{output}(a)
\\
\rulefont{CaOutput'!}&
&(\quorum\echo_1(0)\tand\quorum\echo_1(1))\tnotor \tf{output}(\botval)
\figskip
\figskip
\figunderline{Other rules}
\figskip
\rulefont{CaCorrect}&
&\quorum\correct{\tf{input},\echo_1,\echo_2,\tf{output}} 
\figskip
\rulefont{CaCorrect'}&
&\correct{\tf P}\tor\incorrect{\tf P}
\quad (\tf P{\in}\{\tf{input},\tf{echo}_1,\tf{echo}_2,\tf{output}\})
\figskip
\rulefont{CaInput}&
&(\tf{input}(0)\toplus\tf{input}(1))\tand\tneg\tf{input}(\botval)
\\
\rulefont{CaEcho2_{01}}&
&
\texiaffine \echo_2
\end{array}
$$
As standard, free $a$ above are assumed universally quantified (i.e. there is an invisible $\tall a$ at the start of any axiom with a free $a$).
The axioms above assume predicate symbols $\predsymb =\{\tf{input}, \echo_1, \echo_2, \tf{output} \}$ (as explained in Definition~\ref{defn.logic}).
We use Notation~\ref{nttn.abuse.notation}(\ref{item.abuse.notation.2}) to elide quantified variable symbols in \rulefont{CaEcho2_{01}} and \rulefont{CaEcho2!}. 
\caption{$\ThyCA$: axioms of Declarative Crusader Agreement (Definition~\ref{defn.ca.algorithm})}
\label{fig.ca}
\end{figure}

\begin{figure}
$$
\begin{array}[t]{l@{\quad}r@{\ }l@{\quad}l}
\rulefont{CaAgree}&&
\ThyCA\mentM(\someone\tf{output}(v)\tand\someone\tf{output}(v'))\timpc (v\teq v'\tor v\teq \botval\tor v'\teq\botval)
&\text{(Proposition~\ref{prop.ca.weak.agreement})}
\\
\rulefont{CaValid1}&&
\ThyCA\mentM (\modTB\everyone\tf{input}(v)\tand\tf{output}(v'))\timpc v\teq v'
&\text{(Proposition~\ref{prop.ca.validity})}
\\
\rulefont{CaValid2}&&
\ThyCA\mentM \someone\tf{output}(v)\timpc(\someone\tf{input}(v)\tor v\teq\botval)
&\text{(Proposition~\ref{prop.ca.validity})}
\\
\rulefont{CaLive}&&
\ThyCA\mentM\everyone\texi \tf{output}
&\text{(Proposition~\ref{prop.ca.liveness})}
\end{array}
$$
We use Notation~\ref{nttn.abuse.notation}(\ref{item.abuse.notation.2}) to elide the quantified variable symbol in \rulefont{CaLive}. 
\caption{Correctness properties for $\ThyCA$ (Definition~\ref{defn.ca.goals})}
\label{fig.ca.valid}
\end{figure}

\begin{defn}
\label{defn.ca.goals}
Relevant correctness properties for Crusader Agreement as per~\cite{abraham:ca} are presented in Figure~\ref{fig.ca.valid}.
We consider them in turn:
\begin{enumerate*}
\item
\textbf{Weak Agreement:}
\emph{If two non-faulty parties output values $v$ and $v'$, then either $v=v'$ or at least one of the values is $\botval$.}
\\
This becomes 
$$
\ThyCA\mentM(\someone\tf{output}(v)\tand\someone\tf{output}(v'))\timpc (v\teq v'\tor v\teq \botval\tor v'\teq\botval).
$$
See Proposition~\ref{prop.ca.weak.agreement}.
\item
\textbf{Validity:}
\emph{If all non-faulty parties have the same input, then this is the only possible output \emph{of any non-faulty party}.\footnote{I added the `for any non-faulty party' (the original source material~\cite{abraham:ca} did not say this) because from context it is clear that this is what is intended.  Faulty parties are not constrained.} 
Furthermore, if a non-faulty party outputs $v\neq\botval$, then $v$ was the input of some non-faulty party.}
\\
These become 
$$
\begin{array}{r@{\ }l}
\ThyCA\mentM& (\modTB\everyone\tf{input}(v)\tand\tf{output}(v'))\timpc v\teq v'
\quad\text{and}
\\
\ThyCA\mentM& \someone\tf{output}(v)\timpc(\someone\tf{input}(v)\tor v\teq\botval)
.
\end{array}
$$
See Proposition~\ref{prop.ca.validity}.
\item
\textbf{Liveness:}
\emph{If all non-faulty parties start the protocol then all non-faulty parties eventually output a value.}
\\
This becomes $\ThyCA\mentM\everyone\texi v.\tf{output}(v)$.
See Proposition~\ref{prop.ca.liveness}.
\end{enumerate*}
\end{defn}

\subsection{Discussion of the axioms}

\begin{rmrk}
As for the axioms in Figures~\ref{fig.simple} and~\ref{fig.bb}, the axioms in Figure~\ref{fig.ca} are grouped into \emph{backward rules} with names ending in $?$ which represent reasoning of the form ``if \emph{this} happens, then \emph{that} must have happened''; \emph{forward rules} have names ending with $!$ which represent forward reasoning of the form ``if \emph{this} happens, then \emph{that} must happen''; and \emph{other rules} which reflect particular conditions of the algorithms.

Recall that each axiom is implicitly universally quantified over $a$ if required, and it is valid of a model when it is valid (i.e. it returns $\tvT$ or $\tvB$, but not $\tvF$) when evaluated at every $p\in\points$.
\end{rmrk}

\begin{rmrk}
\label{rmrk.ca.discuss}
We discuss each axiom in Figure~\ref{fig.ca} in turn:
\begin{enumerate*}
\item
\rulefont{CaEcho1?} says: if $\echo_1(v)$ is $\tvT$ at $p$ then $\tf{input}(v)$ is $\tvT$ at some $p'$.

This axiom does not immediately correspond to lines~2 to~4 in Figure~\ref{fig.ca.code} --- contrast with \rulefont{CaEcho1!}, which more clearly corresponds to the algorithm.
We explain and discuss this difference below in Remark~\ref{rmrk.CaEcho?}.
\item
\rulefont{CaEcho2?} says: if $\echo_2(v)$ is $\tvT$ at $p$ then for a quorum of participants $\echo_1(v)$ is $\tvT$ or $\tvB$.
This corresponds to part of lines~5 and~6, since it is clear that if a participant performs $\echo_2(v)$ then it must have seen a quorum of $\echo_1(v)$.

The condition \verb+if didnt send <echo2, *>+ on line~5 is captured by \rulefont{CaEcho2_{01}}, which just asserts that no correct participant may \verb+echo2+ \emph{twice}. 
\item\label{item.ca.error}
\rulefont{CaOutput?} and \rulefont{CaOutput!} reflect lines~7 to~10, 
except for one detail: lines~7 to~10 suggest a pair of axioms
$$
\tf{output}(a) \tnotor (\quorum\echo_2(a)\tand\quorum\echo_1(a))
\qquad
(\quorum\echo_2(a)\tand\quorum\echo_1(a))\tnotor \tf{output}(a)
$$
but instead in Figure~\ref{fig.ca} we have axioms 
$$
\tf{output}(a) \tnotor \quorum\echo_2(a)
\qquad
\quorum\echo_2(a)\tnotor \tf{output}(a) .
$$
Why the difference?

Because, in fact, the right-hand conjunct in the algorithm from~\cite{abraham:ca} is redundant.\footnote{This became clear while doing the axiomatic proofs; the right-hand conjunct was never used.  In fact the same is true in the source article, but this is harder to spot because its reasoning is less explicitly axiomatic.}
We could include it in the axioms and it would do no harm, but it also makes no difference to the proofs and would take up space.

One of the algorithm's authors\footnote{Ittai Abraham; private correspondence.} confirms this: the right-hand conjunct was left in by accident.
It gives stronger correctness properties that are treated in the paper~\cite{abraham:effasa}, but they are not mentioned in the (simplified) online article~\cite{abraham:ca}.
\item
\rulefont{CaOutput'?} asserts that if we output $\botval$ then we must have seen quorums of $\echo_1$ for both $0$ and $1$, as per lines~9 and~10 in Figure~\ref{fig.ca.code}. 
\item
\rulefont{CaCorrect} asserts that there exists a quorum of correct participants, by which we mean precisely that they all return $\tvT$ or $\tvF$ truth-values for $\tf{input}$, $\tf{echo}_1$, $\tf{echo}_2$, and $\tf{output}$, and they never return $\tvB$.
\item
\rulefont{CaCorrect'} asserts that \emph{either} a participant is correct and always returns correct truth-values --- \emph{or} it always returns the faulty truth-value $\tvB$.

The reader might ask why a faulty participant must always return $\tvB$; what if it returns $\tvB$ just sometimes?
The reason is that this logically captures worst-case behaviour.
Specifically, the forward and backward rules are all implications, and it is a fact that if the left-hand side of an implication is $\tvB$, then the implication overall returns $\tvB$ or $\tvT$ and so is valid; the reader can check this in Figure~\ref{fig.3}, just by observing that the row for $\tvB$ in the truth-tables for $\tnotor$ and $\timpc$ is `$\tvT\ \tvB\ \tvB$' and this does not mention $\tvF$.

So, this makes $\tvB$ behave like a `worst case truth-value' in the sense that all forward and backward rules are valid if the left-hand side returns $\tvB$, so that in effect the rule does not constrain behaviour involving faulty participants.
But, we do not need to keep saying `for a non-faulty participant', because the three-valued logic takes care of all this for us, automatically. 
\item
\rulefont{CaInput} expresses line~1: a correct participant inputs either $0$ or $1$.
\item
\rulefont{CaEcho1!} expresses the forward version of lines~2 to 4; echo1 the value of your input and echo1 the value of any contraquorum that you see.
\item
\rulefont{CaEcho2!} expresses the forward version of lines~5 and~6.
We need the existential quantifier because in this case, we only echo2 the value of the first quorum of echo1 that we see. 
\item\label{item.discuss.caoutput}
\rulefont{CaOutput!} was discussed above.
We would just add here that $\modT\echo_2(\botval)$ is impossible by Lemma~\ref{lemm.nice.to.know}(\ref{item.nice.to.know.3}), so that in the presence of the other axioms, we do not need an explicit side-condition in \rulefont{CaOutput!} that $a\neq \botval$.
\item
\rulefont{CaOutput'!} corresponds directly to lines~9 and~10 in Figure~\ref{fig.ca.code} (and is just the reverse of \rulefont{CaOutput'?}). 
\end{enumerate*}
\end{rmrk}

\begin{rmrk}
\label{rmrk.CaEcho?}
\rulefont{CaEcho1?} is rather special.
A direct reading of lines~2 to~4 of the algorithm from Definition~\ref{defn.ca.algorithm} suggests we write the following property:
$$
\echo_1(v) \tnotor (\tf{input}(v) \tor \contraquorum\echo_1(v)) .
$$
Indeed, the forward rule \rulefont{CaEcho1!} looks just like that, just in the other direction.
So why did we not write the predicate above as our \rulefont{CaEcho1?}?

Because it is subtly too weak.
The axiom above would permit a model in which (for example) all participants are correct, and they all input $0$, but they all perform $\echo_1(1)$.
Note that $\echo_1(1)\tnotor \contraquorum\echo_1(1)$ is valid of this model.
Clearly we want to exclude it, and no practical run of the algorithm would generate it, but logically speaking it perfectly well satisfies the restriction above.

What is going on here is that our model has no notion of algorithmic time; there is no notion of state transition system or abstract machine (see the discussion in Subsection~\ref{subsect.algorithmic.time}).
So, our model has no notion of a `first' $\echo_1$ that would need to be justified by an input.

We could now be literal and introduce a notion of time, either as an explicit timestamp parameter to each $\echo_1$, or we could add a notion of time to the modal context in the logic itself.
There would be nothing wrong with that --- it reflects %
what the algorithm actually does 
and is technically entirely feasible.

But we will prefer a simpler, more general, and arguably more elegant approach.\footnote{But as always: what you axiomatise depends on what you want to represent.  Today, we might make one design decision about how much detail we want to include in our model.  Tomorrow, or for a different protocol, or for a different intended audience, it would be perfectly legitimate to choose different tradeoffs.}
To explain how this works, we reason as follows:
we observe of the algorithm that if a correct participant performs $\echo_1(v)$ then either \emph{it} has input $v$, or $f\plus 1$ participants --- and thus in particular one correct participant --- also have performed $\echo_1(v)$.
By an induction on time, we see that if a correct participant performs $\echo_1(v)$ then some correct participant has input $v$.
This justifies writing \rulefont{CaEcho1?} as we do in Figure~\ref{fig.ca}.

It is not immediately obvious that this suffices to prove our correctness properties; we have to check that in the proofs. 
And, it turns out that the proofs work well, which gives a formal sense in which the \rulefont{CaEcho1?} axiom in Figure~\ref{fig.ca} is correct and in fact, it is more general than the more literal rendering.\footnote{Why is a weaker axiom better?  Surely strong axioms with lots of structure are good?  No: generally speaking, the game is to derive strong and well-structured correctness properties from weak and minimally-structured axioms; our game is to derive the strongest possible theorems from the weakest possible assumptions.} 
\end{rmrk}

\subsection{Proofs of correctness properties for Crusader Agreement}

\begin{rmrk}
\label{rmrk.fix.model.ca}
As we did for Bracha Broadcast in Remark~\ref{rmrk.fix.model.bb}, for the rest of this section, we fix a model (Definition~\ref{defn.sem}) 
$\mu=(\Val,(\points,\opens),\varsigma)$ 
such that the underlying semitopology $(\points,\opens)$ is 3-twined (Definition~\ref{defn.3twined}), so that we have Theorem~\ref{thrm.3twined.logic} and Corollary~\ref{corr.3twined.cologic}.
\end{rmrk}

\subsubsection{Weak Agreement}

\begin{rmrk}
The proof of Proposition~\ref{prop.ca.weak.agreement} is straightforward. 
Note how the `non-faulty' precondition is integrated into the use of $\timpc$, which assumes that the left-hand side returns $\tvT$ (the non-faulty truth-value).\footnote{On the right-hand side, $\model{v\teq v'}\in\{\tvT,\tvF\}$, so equality always returns a correct truth-value.}
\end{rmrk}

\begin{prop}[Weak Agreement]
\label{prop.ca.weak.agreement}
Suppose $\mentM\ThyCA$ and $v,v'\in\{0,\botval,1\}$.
Then 
$$
\mentM(\someone\tf{output}(v)\tand\someone\tf{output}(v'))\timpc (v\teq v'\tor v\teq \botval\tor v'\teq\botval)
.
$$
This formalises the following informal statement:
\begin{quote}
``If two non-faulty parties output values $v$ and $v'$, then either $v=v'$ or one of the values is $\botval$''
\end{quote}
\end{prop}
\begin{proof}
Using \strongmodusponens it suffices to assume $\mentM\modT\someone\tf{output}(v)\tand\modT\someone\tf{output}(v')$ and $v,v'\neq\botval$, and prove $v=v'$.

So suppose $\mentM\modT\someone\tf{output}(v)\tand\modT\someone\tf{output}(v')$ and $v,v'\neq\botval$.
By \rulefont{CaOutput?} $\mentM \quorum\echo_2(v)\tand\quorum\echo_2(v')$.
By \rulefont{CaCorrect} $\mentM\quorum\correct{\echo_2}$.
By Theorem~\ref{thrm.3twined.logic} and Lemma~\ref{lemm.semi.char}(\ref{item.semi.char.TF}) $\mentM\modT\someone(\echo_2(v)\tand\echo_2(v'))$.
By \rulefont{CaEcho2_{01}} and Proposition~\ref{prop.unique.affine.existence}(\ref{item.unique.affine.existence.01}.\ref{item.unique.affine.existence.01.a}\&\ref{item.unique.affine.existence.01}.\ref{item.unique.affine.existence.01.e}) $v=v'$ as required. 
\end{proof}

\begin{xmpl}
\label{xmpl.output.twice}
$\tf{output}$ is not functional, by which we mean that $\modelm{\texiaffine v.\tf{output}(v)}$ need not be equal to $\tvT$. 
We give an example to show that $p\mentM\modT(\tf{output}(\botval) \tand \tf{output}(1))$ is possible.

Consider a model with ten participants; all participants are correct; five participants input $0$ and five input $1$; quorums are any set with at least seven participants and thus contraquorums are any set with at least four participants.
The reader can check that every participant necessarily produces $\echo_1(0)$ and $\echo_1(1)$, so that all participants $\tf{output}(\botval)$.
It is also consistent with the axioms for every participant to produce $\echo_2(1)$ and $\tf{output}(1)$. 
\end{xmpl}

\subsubsection{Validity}

We work towards Proposition~\ref{prop.ca.validity}; we will need Lemma~\ref{lemm.ca.inputunique}:
\begin{lemm}
\label{lemm.ca.inputunique}
Suppose $\mentM\ThyCA$ and $p\in\points$.
Then:
\begin{enumerate*}
\item\label{item.ca.inputunique.0}
$p\mentM\tneg\tf{input}(\botval)$.\footnote{This reflects line~1 in Figure~\ref{fig.ca.code}.}
\item\label{item.ca.inputunique.1}
$p\mentM\tf{input}(v)$ if and only if $p\mentM\tneg\tf{input}(1\minus v)$, for $v\in\{0,1\}$.
\item\label{item.ca.inputunique.2}
$p\mentM\modT\tf{input}(v)$ and $p\mentM\modT\tf{input}(v')$ implies $v=v'$, for $v,v'\in\{0,\botval,1\}$.
\item\label{item.ca.inputunique.3}
$p\mentM\everyone\tf{input}(v)$ and $p\mentM\modT\someone\tf{input}(v')$ implies $v=v'$, for $v,v'\in\{0,\botval,1\}$.
\item\label{item.ca.inputunique.4}
$p\mentM\tneg\echo_2(\botval)$.
\end{enumerate*}
\end{lemm}
\begin{proof}
We consider each part in turn:
\begin{enumerate}
\item
We consider two subcases:
\begin{itemize*}
\item
If $\modelm{\tf{input}(\botval)}(p)\in\{\tvF,\tvB\}$ then $p\mentM\tneg\tf{input}(\botval)$ and we are done.
\item
If $\modelm{\tf{input}(\botval)}(p)=\tvT$
then
$\modelm{\tneg\tf{input}(\botval)}(p)=\tvF$, contradicting the right-hand conjunct of \rulefont{CaInput} (that $\tneg\tf{input}(\botval)$). 
So this subcase is impossible.
\end{itemize*}
\item
We consider three subcases:
\begin{itemize*}
\item
\emph{Suppose $\modelm{\tf{input}(v)}(p)=\tvB$.}

Then by \rulefont{CaCorrect'} (for $\tf{input}$) $\modelm{\tf{input}(v')}(p)=\tvB$ for \emph{every} $v'\in\{0,\botval,1\}$ and in particular $\modelm{\tf{input}(1\minus v)}(p)=\tvB=\tneg\tvB$ so we are done.
\item
\emph{Suppose $\modelm{\tf{input}(v)}(p)=\tvT$.}

Then by \rulefont{CaCorrect'} (for $\tf{input}$) $\modelm{\tf{input}(v')}(p)\in\{\tvT,\tvF\}$ for \emph{every} $v'\in\{0,\botval,1\}$ and in particular $\modelm{\tf{input}(1\minus v)}(p)\in\{\tvT,\tvF\}$.
By \rulefont{CaInput} and routine arguments on the truth-tables from Figure~\ref{fig.3},
$\tf{input}(v')$ is true for precisely \emph{one} value $v'\in\{0,1\}$.
The result follows.
\item
The case that $\modelm{\tf{input}(v)}(p)=\tvF$ is precisely symmetric to the case of $\tvT$.
\end{itemize*}
\item
By routine logical reasoning from parts~\ref{item.ca.inputunique.0} and~\ref{item.ca.inputunique.1} of this result.
\item
Suppose $\mentM\everyone\tf{input}(v)$ and $p\in\points$ and $p\mentM\modT\tf{input}(v')$.
Then $p\mentM\tf{input}(v)$ and $p\mentM\modT\tf{input}(v')$.
By \rulefont{CaCorrect'} (for $\tf{input}$) and Proposition~\ref{prop.tand.tor}(\ref{item.para.ment}), $p\mentM\modT\tf{input}(v)$.
We use part~\ref{item.ca.inputunique.2} of this result.
\item
If $\modelm{\echo_2(\botval)}(p)\in\{\tvF,\tvB\}$ then $p\mentM\tneg\echo_2(\botval)$ and we are done.
So suppose $\modelm{\echo_2(\botval)}(p)=\tvT$, i.e. $p\mentM\modT\echo_2(\botval)$; we will derive a contradiction.

By \weakmodusponens and \rulefont{CaEcho2?} $\mentM\quorum\echo_1(\botval)$.
By \rulefont{CaCorrect} $\mentM\quorum\correct{\echo_1}$ so using Theorem~\ref{thrm.3twined.logic} and Lemma~\ref{lemm.semi.char}(\ref{item.semi.char.2}) (or just direct from the 3-twined property in Definition~\ref{defn.3twined}) $\mentM\someone\modT\echo_1(\botval)$. 
By \strongmodusponens and \rulefont{CaEcho1?} $\mentM\someone\modT\tf{input}(\botval)$.
But this contradicts part~\ref{item.ca.inputunique.0} of this result.
\qedhere\end{enumerate}
\end{proof}

\begin{prop}[Validity]
\label{prop.ca.validity}
\leavevmode
Suppose $\mentM\ThyCA$ and $v\in\{0,\botval,1\}$.
Then:
\begin{enumerate*}
\item\label{item.ca.validity.1}
$\mentM \someone\tf{output}(v)\timpc(\someone\tf{input}(v)\tor v\teq\botval)$. 

``If a non-faulty party outputs $v\neq\botval$, then $v$ was the input of some non-faulty party.''
\item\label{item.ca.validity.1b}
$\mentM \someone\tf{output}(\botval)\timpc(\someone\tf{input}(0)\tand\someone\tf{input}(1))$ (we will need this to prove part~\ref{item.ca.validity.2}).
\item\label{item.ca.validity.2}
$\mentM \everyone\tf{input}(v)$ and $\mentM\modT\someone\tf{output}(v')$ implies $v\teq v'$, for every $v'\in\{0,\botval,1\}$.

``If all non-faulty parties have the same input, then this is the only possible output \emph{of any non-faulty party}.''
\end{enumerate*}
\end{prop}
\begin{proof}
We consider each part in turn:
\begin{enumerate}
\item
By \strongmodusponens it suffices to show that $\mentM\modT\someone\tf{output}(v)$ and $v\in\{0,1\}$ implies $\mentM\modT\someone\tf{input}(v)$.

So suppose $\mentM\modT\someone\tf{output}(v)$ and $v\in\{0,1\}$.
By \rulefont{CaOutput?} (since $v=0$ or $v=1$) $\mentM\quorum\echo_2(v)$ and by \rulefont{CaCorrect} $\mentM\quorum\correct{\echo_2}$.
Using Theorem~\ref{thrm.3twined.logic} and Lemma~\ref{lemm.semi.char}(\ref{item.semi.char.TF}) (or just direct from the 3-twined property in Definition~\ref{defn.3twined}) $\mentM\modT\someone\echo_2(v)$. 

By similar reasoning using \rulefont{CaEcho2?} in place of \rulefont{CaOutput?}, we conclude that $\mentM\modT\someone\echo_1(v)$. 
By \rulefont{CaEcho1?} $\mentM\modT\someone\tf{input}(v)$ as required.
\item
By \strongmodusponens it suffices to show that $\mentM\modT\someone\tf{output}(\botval)$ implies $\mentM\modT\someone\tf{input}(0)$ and $\mentM\modT\someone\tf{input}(1)$.

So suppose $\mentM\modT\someone\tf{output}(\botval)$. 
By \rulefont{CaOutput'?} $\mentM\quorum\echo_1(0)\tand\quorum\echo_1(1)$.

We continue to reason as for part~\ref{item.ca.validity.1} of this result to conclude $\mentM\modT(\someone\tf{input}(0)\tand\someone\tf{input}(1))$ as required.
\item
Suppose $\mentM \everyone\tf{input}(v)$ and suppose $\mentM\modT\someone\tf{output}(v')$; we will show $\mentM\modT(v\teq v')$, i.e. $v=v'$.
There are now two sub-cases:
\begin{itemize*}
\item
\emph{Suppose $v'\in\{0,1\}$.}\quad
By part~\ref{item.ca.validity.1} of this result $\mentM\modT\someone\tf{input}(v')$.
But we assumed $\mentM\everyone\tf{input}(v)$, so by Lemma~\ref{lemm.ca.inputunique}(\ref{item.ca.inputunique.3}) $v=v'$ as required.
\item
\emph{Suppose $v'=\botval$.}
By part~\ref{item.ca.validity.1b} of this result $\mentM\modT\someone\tf{input}(0)$ and $\mentM\modT\someone\tf{input}(1)$.
But we assumed $\mentM\everyone\tf{input}(v)$, so by Lemma~\ref{lemm.ca.inputunique}(\ref{item.ca.inputunique.3}) $0=v=1$, a contradiction. 
So this case is impossible.
\end{itemize*} 
\emph{(A word on notation: \rulefont{CaValid1} in Figure~\ref{fig.ca.valid} is written ``\,$\mentM (\modTB\everyone\tf{input}(v)\tand\tf{output}(v'))\timpc v\teq v'$''.
It is a fact that this is equivalent to ``\,$\mentM \everyone\tf{input}(v)$ and $\mentM\modT\someone\tf{output}(v')$ implies $v\teq v'$'', so the statement of this part of the Lemma is accurate; we have just rephrased it for convenience.)} 
\qedhere\end{enumerate}
\end{proof}

\begin{rmrk}
\label{rmrk.correctness.correct.1}
We take a minute to discuss the precise form of the validity properties in Proposition~\ref{prop.ca.validity}.
Consider this English sentence:
\begin{quote}
\emph{``If a non-faulty party outputs $v\neq\botval$, then $v$ was the input of some non-faulty party.''}
\end{quote}
We render this in \rulefont{CaValid2} as
$$
\ThyCA\mentM \someone\tf{output}(v)\timpc(\someone\tf{input}(v)\tor v\teq\botval) .
$$
Using \strongmodusponens, this means that if $v\neq\botval$ and $\modelm{\tf{output}(v)}(p)=\tvT$ for some $p\in\points$ then $\modelm{\tf{input}(v)}(p')=\tvT$ for some $p'\in\points$.
We are using $\tvT$ and $\tvF$ as the truth-values of correct (non-faulty) participants, and $\tvB$ as the truth-value for incorrect (faulty) participants, so this corresponds precisely to the English sentence quoted above.

It might be instructive to consider a subtly incorrect rendering of \rulefont{CaValid1}  
\begin{quote}
\emph{``If all non-faulty parties have the same input, then this is the only possible output \emph{of any non-faulty party}''}
\end{quote}
as
$$
\ThyCA\mentM \quorum\tf{input}(v)\timpc \tf{output}(v) .
$$
This is wrong, for two reasons: 
\begin{enumerate*}
\item
$\tf{output}$ is just a predicate-symbol; it is not a function-symbol.
As such, $\tf{output}(v)$ does not \emph{a priori} have to return $\tvT$ or $\tvB$ on just one value $v$.\footnote{Indeed, in the case of this particular protocol $\tf{output}$ can be $\tvT$ on two values; see Example~\ref{xmpl.output.twice}.} 
$\tf{output}(v)$ could be $\tvT$ and also $\tf{output}(v')$ might be $\tvT$ for some other $v'$.
\item
More subtly, just because $\mentM\modT\quorum\tf{input}(v)$ holds does not \emph{a priori} mean that $\mentM\modT\someone\tf{input}(1\minus v)$ cannot hold.
Remember that we are working over an abstract 3-twined semitopology; there might be a quorum of correct participants who $\tvT$-do $\tf{input}(v)$ and \emph{also} some other correct participant who $\tvT$-does $\tf{input}(v')$ for some other $v'$.
\end{enumerate*}
This is why we write $\modTB\everyone\tf{input}(v)$ in \rulefont{CaValid1} in Figure~\ref{fig.ca.valid}: it means that every participant is either faulty or, if it is not faulty, it inputs $v$.
The use of exclusive-or $\toplus$ in \rulefont{CaInput} ensures that $\tf{input}$ is functional, so this does indeed render the idea of the English sentence quoted above.
\end{rmrk}

\subsubsection{Liveness}

We work towards Proposition~\ref{prop.ca.liveness}.
We will need Lemma~\ref{lemm.some.echo2.quorum.echo1} and two corollaries of it:
\begin{lemm}
\label{lemm.some.echo2.quorum.echo1}
Suppose $\mentM\ThyCA$ and $v\in\{0,\botval,1\}$.
Then:
\begin{enumerate*}
\item
$\mentM\quorum\echo_1(v)$ implies $\mentM\modT\contraquorum\echo_1(v)$.
\item
$\mentM\modT\contraquorum\echo_1(v)$ implies $\mentM\everyone\echo_1(v)$.
\item\label{item.some.echo2.quorum.echo1.3}
$\mentM\everyone\echo_1(v)$ implies $\mentM\modT\quorum\echo_1(v)$.
\end{enumerate*}
\end{lemm}
\begin{proof}
Suppose $\mentM\quorum\echo_1(v)$.
By \rulefont{CaCorrect} (for $\echo_1$) and Theorem~\ref{thrm.3twined.logic} $\mentM\modT\contraquorum\echo_1(v)$.
By \rulefont{CaEcho1!} (right-hand disjunct, for $\contraquorum\echo_1$) $\mentM\everyone\echo_1(v)$.
By \rulefont{CaCorrect} (for $\echo_1$) and Lemma~\ref{lemm.everyone.and.quorum}(\ref{item.everyone.and.quorum.2}) $\mentM\modT\quorum\echo_1(v)$.
\end{proof}

\begin{corr}
\label{corr.some.echo2.quorum.echo1}
Suppose $\mentM\ThyCA$ and $v\in\{0,\botval,1\}$.
Then $\mentM\echo_2(v)\timpc\quorum\echo_1(v)$.
\end{corr}
\begin{proof}
Suppose $p\mentM\modT\echo_2(v)$.
By \rulefont{CaEcho2?} $\mentM\quorum\echo_1(v)$.
By Lemma~\ref{lemm.some.echo2.quorum.echo1} %
$\mentM\modT\quorum\echo_1(v)$.
\end{proof}

\begin{corr}
\label{corr.contraquorum.input}
Suppose $\mentM\ThyCA$.
Then:
\begin{enumerate*}
\item\label{item.contraquorum.input.1}
$\mentM \modT\contraquorum(\tf{input}(0)\tand\correct{\echo_1}) \tor \modT\contraquorum(\tf{input}(1)\tand\correct{\echo_1})$.
\item\label{item.contraquorum.input.2}
$\mentM \modT\contraquorum\echo_1(0) \tor \modT\contraquorum\echo_1(1)$.
\item\label{item.contraquorum.input.4}
$\mentM \modT\quorum\echo_1(0) \tor \modT\quorum\echo_1(1)$.
\item\label{item.contraquorum.input.5}
$\mentM \everyone(\echo_2(0) \tor \echo_2(1))$.
\item\label{item.contraquorum.input.6}
$\mentM \modT\quorum(\echo_2(0) \tor \echo_2(1))$.
\end{enumerate*}
\end{corr}
\begin{proof}
We consider each part in turn:
\begin{enumerate}
\item
Using \rulefont{CaInput} $\mentM \everyone\bigl(\tf{input}(0)\tor\tf{input}(1)\bigr)$, and using \rulefont{CaCorrect} $\mentM\quorum\correct{\echo_1,\tf{input}}$.
Thus by Lemma~\ref{lemm.everyone.and.quorum}(\ref{item.everyone.and.quorum.1})
$\mentM\quorum\bigl((\tf{input}(0)\tor\tf{input}(1))\tand\correct{\echo_1,\tf{input}}\bigr)$,
and rearranging we obtain
$$
\mentM\quorum((\tf{input}(0)\tand\correct{\echo_1,\tf{input}})\tor(\tf{input}(1)\tand\correct{\echo_1,\tf{input}})).
$$
By Corollary~\ref{corr.3twined.cologic} (since we assumed in Remark~\ref{rmrk.fix.model.ca} that $(\points,\opens)$ is 3-twined) 
$$
\mentM\contraquorum(\tf{input}(0)\tand\correct{\echo_1,\tf{input}})\tor\contraquorum(\tf{input}(1)\tand\correct{\echo_1,\tf{input}}).
$$
By routine calculations from Figure~\ref{fig.3} we conclude that
$$
\mentM \modT\contraquorum(\tf{input}(0)\tand\correct{\echo_1}) \tor \modT\contraquorum(\tf{input}(1)\tand\correct{\echo_1}) 
$$
as required.
\item
We combine part~\ref{item.contraquorum.input.1} of this result with \rulefont{CaEcho1!} (left-hand disjunct, for $\tf{input}$) using \weakmodusponensnoref and Proposition~\ref{prop.tand.tor}(\ref{item.para.ment}).
\item
From part~\ref{item.contraquorum.input.2} of this result using Lemma~\ref{lemm.some.echo2.quorum.echo1}. 
\item
By part~\ref{item.contraquorum.input.4} of this result and \rulefont{CaEcho2!} we have $\mentM\everyone\texi v.\echo_2(v)$.
By Lemma~\ref{lemm.ca.inputunique}(\ref{item.ca.inputunique.4}) $\mentM\tneg\echo_2(\botval)$.
The result follows using \rulefont{CaCorrect'} (for $\echo_2$).
\item
From part~\ref{item.contraquorum.input.5} of this result using \rulefont{CaCorrect} (for $\echo_2$) and Lemma~\ref{lemm.everyone.and.quorum}(\ref{item.everyone.and.quorum.2}).
\qedhere\end{enumerate}
\end{proof}

\begin{prop}[Liveness]
\label{prop.ca.liveness}
If $\mentM\ThyCA$ then $\mentM\everyone\texi v.\tf{output}(v)$.

``If all non-faulty parties start the protocol then all non-faulty parties eventually output a value.''
\end{prop}
\begin{proof}
By Corollary~\ref{corr.contraquorum.input}(\ref{item.contraquorum.input.6}) there exists a quorum $O\in\opensne$ such that 
$$
\Forall{p\in O}(p\mentM\modT(\tf{echo}_2(0)\tor\tf{echo}_2(1))) .
$$ 
Note in passing that by \rulefont{CaEcho2_{01}} and Proposition~\ref{prop.unique.affine.existence}(\ref{item.unique.affine.existence.01}.\ref{item.unique.affine.existence.01.e}), these two disjuncts are mutually exclusive.
There are now two cases:
\begin{enumerate*}
\item
Suppose $\Forall{p\in O}(p\mentM\modT\tf{echo}_2(v))$ for some $v\in\{0,1\}$, so that $\mentM\modT\quorum\echo_2(v)$.

Then using \rulefont{CaOutput!} $\mentM\everyone\tf{output}(v)$.
\item
Suppose there exist $p,p'\in O$ such that $p\mentM\modT\tf{echo}_2(0)$ and $p'\mentM\modT\tf{echo}_2(1)$, so that $\mentM\modT\someone\tf{echo}_2(0) \tand \modT\someone\tf{echo}_2(1)$.

By Corollary~\ref{corr.some.echo2.quorum.echo1} 
$\mentM\modT\quorum\echo_1(0)$ and $\mentM\modT\quorum\echo_1(1)$
and by \rulefont{CaOutput'!} $\mentM\everyone\tf{output}(\botval)$.
\end{enumerate*}
In either case, $\mentM\everyone\texi v.\tf{output}(v)$, as required.
\end{proof}

\begin{rmrk}
\label{rmrk.correctness.correct.2}
We continue the theme of Remark~\ref{rmrk.correctness.correct.1} and comment on some design subtleties of the logical correctness property for Liveness:
\begin{enumerate*}
\item
This property for Liveness would be subtly incorrect:
$$
\ThyCA \mentM \quorum\texi v.\tf{output}(v).
$$
It might look right, but we are working over an abstract 3-twined semitopology so it might be that there is a quorum of correct participants who output some value, and \emph{also} some other correct participant not in that quorum, who does not.
\item
The English statement of Liveness includes a proviso `if all non-faulty parties start the protocol', which is omitted in $\everyone\texi v.\tf{output}(v)$ (and also in $\quorum\texi v.\tf{output}(v)$).
This is for two reasons: first, if a party does not start the protocol then arguably from our point of view it is faulty, so the proviso trivially holds by definition; but second, our model has no notion of time so talking about when and whether a party starts or does not start doing something is just not in the scope of our analysis.
\item
Does $\everyone\texi v.\tf{output}(v)$ not mean that \emph{everyone} outputs a value, including faulty parties?
A faulty participant $p$ will return $\tvB$ for $\tf{output}(v)$ for any $v$, which makes $\modelm{\texi v.\tf{output}(v)}(p)=\tvB$.
This is a valid truth-value.

So $\everyone\texi v.\tf{output}(v)$ elegantly and accurately captures the idea that \emph{`every correct/non-faulty participant outputs a value'}.
\end{enumerate*}
\end{rmrk}

\subsection{Further discussion of the axioms and proofs}

\begin{rmrk}
Compare \rulefont{CaEcho1?} from Figure~\ref{fig.ca} with \rulefont{BrEcho?} from Figure~\ref{fig.bb}.
Note how similar they are: the only difference, in fact, is that \rulefont{CaEcho1?} uses the strong implication $\timpc$ whereas \rulefont{BrEcho?} uses the weak implication $\tnotor$.
This illustrates a pattern: treating distributed algorithms as axiom-systems abstracts them, and this reveals their structure, sometimes in non-obvious ways. 
\end{rmrk}

In Remark~\ref{rmrk.ca.discuss}(\ref{item.discuss.caoutput}) we mentioned that \rulefont{CaOutput!} does not have an explicit side-condition that $a\neq\botval$, because this is a Lemma.
As promised, here is the statement and proof: 
\begin{lemm}
\label{lemm.nice.to.know}
Suppose $\mentM\ThyCA$ and $v\in\{0,\botval,1\}$.
Then:
\begin{enumerate*}
\item\label{item.nice.to.know.1}
$\mentM\tf{input}(v) \timpc (v\teq 0\tor v\teq 1)$.
\item\label{item.nice.to.know.2}
$\mentM\tf{echo}_1(v) \timpc (v\teq 0\tor v\teq 1)$.
\item\label{item.nice.to.know.3}
$\mentM\tf{echo}_2(v) \timpc (v\teq 0\tor v\teq 1)$.
\end{enumerate*}
\end{lemm}
\begin{proof}
We consider each part in turn, freely using \strongmodusponens at each $p\in\points$:
\begin{enumerate}
\item
Suppose $p\mentM\modT\tf{input}(v)$.
Then $v=0\lor v=1$ is direct from \rulefont{CaInput}.
\item
Suppose $p\mentM\modT\tf{echo}_1(v)$.
By \rulefont{CaEcho1?} $p\mentM\modT\someone\tf{input}(v)$, and we use part~\ref{item.nice.to.know.1} of this result.
\item 
Suppose $p\mentM\modT\tf{echo}_2(v)$.
By \rulefont{CaEcho2?} $\mentM\quorum\tf{echo}_1(v)$, and by \rulefont{CaCorrect} (for $\tf{echo}_1$) and Theorem~\ref{thrm.3twined.logic} $\mentM\modT\contraquorum\tf{echo}_1(v)$.
By Lemma~\ref{lemm.semi.char}(\ref{item.semi.char.2}) $\mentM\modT\someone\tf{echo}_1(v)$ and we use part~\ref{item.nice.to.know.2} of this result.
\qedhere\end{enumerate}
\end{proof}

\begin{rmrk}
There are some subtle differences in the treatment of correctness between \ThyBB and \ThyCA.
In \ThyBB we insist that $\ready$ and $\echo$ are correct on quorums (= nonempty open sets), but we do not insist that these quorums must be equal.
In contrast, in \ThyCA we insist that there is a single quorum on which $\tf{input}$, $\echo_1$, $\echo_2$, and $\tf{output}$ are all correct.

Why the difference?
The correctness properties for \ThyCA in Definition~\ref{defn.ca.goals} use a notion of `non-faulty participant' as a participant that is correct throughout a run of the algorithm, and this is explicitly used when Liveness talks about `all non-faulty parties' at the start and at the end of a run of the algorithm.
\ThyBB mentions correct participants too, of course, but the relevant correctness properties from Remark~\ref{rmrk.bb.goals} are subtly more lax in that they can be phrased just in terms of correctness \emph{at a particular predicate}, rather than correctness throughout the run.
The axioms reflect this.

It might be possible to make slightly weaker correctness assumptions of \ThyCA: it would suffice to just check every mention of \rulefont{CaCorrect} in the proofs, and see whether a weaker axiom would work that (for example) does not assume the same quorum for all predicates.
We would then just have to carefully define what we mean by a `faulty' participant. 
We leave this for future work.

Such questions reflect the fact that logic is doing its job, by making things explicit and precise.
In English it is easy to write `faulty participant' or `non-faulty participant', but what this actually \emph{means} can be, and often is, left implicit (or left to the reader to fill in from background culture and expertise).
In the logical world, such concepts can be nailed down using axioms and given precise meaning; then design decisions can be discussed --- using modal logic, English, or some natural combination of the two --- within this logical framework.
\end{rmrk}

\section{Conclusions and related and future work}
\label{sect.conclusions}

\subsection{Conclusions}

We have illustrated how a distributed algorithm can be represented declaratively via axioms.
Slogans are:
\begin{enumerate*}
\item
A distributed protocol = a logical theory (a set of axioms) in a modal logic.
\item
A run of the protocol = a model of that modal logic theory.
\item
A possible world\footnote{`Possible world' is logician's jargon for a point in the model.  Recall that our logic is modal, which means that there is a collection of \emph{possible worlds}, and predicates take truth-values \emph{at these worlds.}} 
= a participant.\footnote{Note that a possible world is \emph{not} a run of the protocol.  In this paper possible worlds are identified with participants; in~\cite{gabbay:decasd} possible worlds consist of a pair of a participant and a round number; in~\cite{gabbay:hettrb} possible worlds are more complex, but still feature a participant.  
So perhaps we can refine this slogan to say in general terms that a possible world consists of a participant, plus other relevant data as required by the protocol being axiomatised.  The precise structure of this data (if any) can give interesting abstract insights into the notion of logical time (if any) of the protocol.}
\item
Correctness properties of the protocol = properties of the class of all models of the theory.
\end{enumerate*} 
We use three-valued logic, such that the third truth-value helps us to represent byzantine behaviour; and we use semitopologies to represent quorums; and then we exploit the close connection between (semi)topology and logic to capture a declarative essence of the distributed algorithm.

The result is simple, yet powerful. 
It is mathematically novel -- three-valued modal logic over a semitopology is new, as is (to the best of my knowledge) declarative axiomatic formal specification of distributed protocols -- and it illuminates the algorithm in informative ways.

It is also practically worthwhile:
a high-level formal specification in logical style is useful to inform new implementations, for modularity, for updatability, to promote diversity of implementations, for verifying implementation correctness, for provable security, for analysis and formal checking of security properties, and for more.
It keeps only what is essential about the object of study, and behaves as an abstract reference implementation of the logical essence of what it \emph{is}.
Any implementation can legitimately claim to be `an implementation of Bracha Broadcast', and will necessarily enjoy the correctness properties of such, if the models that it generates satisfy the axioms in Figure~\ref{fig.bb}.\footnote{Recall for example the evolution to taking such specifications as standard practice in the world of blockchain tokens: e.g. the ERC20 token standard \href{https://eips.ethereum.org/EIPS/eip-20}{eips.ethereum.org/EIPS/eip-20} or the FA2 token standard \href{https://archetype-lang.org/docs/templates/fa2/}{archetype-lang.org/docs/templates/fa2/}.  The point here is cultural (not technical): it is now generally recognised that a logical specification of a token is essential.  Standard formal verification tools can then be applied~\cite{gabbay_et_al:OASIcs.FMBC.2021.2}.}

The examples in this paper are simple by the standards of modern distributed algorithms.
This is because this paper focuses on the underlying technique: the more complex the algorithm we apply these ideas to, the \emph{more favourable} the comparison with its declarative specification is likely to become, because logic helps us to control complexity, elide detail, and focus on essential system properties.
We will return to this point later when we discuss a non-trivial application of the declarative method (to \emph{Heterogeneous Paxos}) in Future Work below.

But here, let us review once more the toy examples in this paper on their own terms.
I would argue that the axiomatic approach provides a step up in mathematical abstraction and rigour; the modal logic  approach to quorums and contraquorums is conceptually clean and free of counting or combinatorial arguments; and the use of a third byzantine truth-value means that byzantine behaviour comes `for free', in the sense that side-conditions of the form `for a correct participant' can be elided because they are taken care of automagically by the logic.
All of this leads, even for simple examples, to compact and precise axioms.

Note that we even found an error (an unintended redundancy) in the pseudocode presentation of Crusader Agreement, as discussed in Remark~\ref{rmrk.ca.discuss}(\ref{item.ca.error}).
This was not obvious from reading the pseudocode and the English discussion of the source material.
However, once this material was rendered into declarative axioms and treated logically, the redundancy became evident.
So even on examples specifically chosen for their simplicity, the declarative analysis uncovered something new and unexpected. 
One wonders what else we could find, understand, optimise, and even correct if we applied this technique to larger, newer, less well-studied algorithms.

\subsection{Related work}
\label{subsect.related.work}

We can think of the logic of this paper as a declarative complement to TLA+, which (while also a modal logic) is more imperative in flavour.
TLA+ is an industry standard %
which is widely and successfully used to specify and verify implementations of distributed protocols~\cite{DBLP:journals/toplas/Lamport94,DBLP:books/aw/Lamport2002,DBLP:conf/charme/YuML99}. 
However, TLA+ is a tool for specifying transition systems.
This is therefore in the spirit of imperative programming, and this paper complements that technique with a declarative approach.

Declarative specification itself is well-studied and its benefits are well-understood: this is what makes (for example) Haskell and Lisp different from C and Bash.
What this paper does is show how to bring such principles to bear on a new class of distributed algorithms, of which Bracha Broadcast is a simple but canonical representative. 
Perhaps axioms like those in this paper might one day be executed as logic clauses in a suitable declarative programming framework, but in the first instance their natural home is likely to be as a source of truth in a specification document, in a formal verification in a theorem-prover, or in a model-checker. 

In practice, including in industrial practice, broadcast and consensus protocols are typically specified in English (as in Example~\ref{xmpl.vote} or Remark~\ref{rmrk.high-level.bb}) or in pseudocode (as in Figure~\ref{fig.ca.code}).
The webpage~\cite{abraham:ca} or the book~\cite{cachinbook} are good examples of the genre.
If formalisation is undertaken, it proceeds in TLA+ (or possibly in LEAN or a similar system), but what gets formalised is still basically a (larger or smaller) description of a (larger or smaller) concrete transition relation.

The essential feature of the logical approach in this paper is to treat the algorithm declaratively at a higher level of abstraction, as an axiomatic theory, rather than as a transition system.
It is not obvious that this should preserve the essence of the algorithm, but it does.

This naturally leads to a question: suppose we are given a declaratively specified protocol `X', and some (pseudo-)code implementation `C' that claims to satisfy specification X.
How could we prove that indeed code C does satisfy specification X?
This may not be trivial (depending on the complexity of C and X), but it would just amount to lemmas proving (formally or informally, depending on the level of assurance that we want) that C satisfies X.
This is not a novel scenario: in many other fields it is standard to have a concrete implementation, and an abstract (axiomatic) specification, and we wish to prove that the former instantiates the latter.
Indeed one instance of this arose in my own work on cryptocurrency token implementations (as the reader can imagine, the design of these is safety-critical)~\cite{gabbay_et_al:OASIcs.FMBC.2021.2}; the reader is referred in particular to the methodological discussion in Remark~11 and Figure~8(b) of that paper, and in particular the benefits of this methodology for proving correctness properties once-and-for-all from the axioms, and then reusing those axiom-based proofs for multiple implementations just by updating the relevant implementation-to-axiom lemmas.
While the work involved was not trivial, in technical terms it was routine, and things \emph{like} this are done in industrial research every day.

\subsection{Unrelated work}

There are fields adjacent to, or superficially similar to, distributed algorithms, and many applications of logic and topology.
There is a danger of being misled by a keyword like `topology' to imagine that this paper is doing one thing, when it is doing another.

So it might be useful to spell out what this paper is not.
\begin{enumerate}
\item
\emph{This is not a paper on algebraic topology.}\quad
We mention this because algebraic topology has been applied to the solvability of distributed-computing tasks in various computational models (e.g. the impossibility of wait-free $k$-set consensus using read-write registers and the Asynchronous Computability Theorem; a good survey is in~\cite{herlihy_distributed_2013}).
Semitopology is topological in flavour, but it is not the same as topology, and this is not a paper about algebraic topology applied to the solvability of distributed-computing tasks.
\item
\emph{This is not a paper about large examples or difficult theorems.}\quad
The examples in this paper (voting, broadcast, agreement) are toy and/or textbook algorithms.
This is a feature, not a bug! 
The point of this approach is to attain simplicity via abstraction in logic.
So for example: if the proof of Proposition~\ref{prop.no.duplication} looks boringly similar to the proof of Proposition~\ref{prop.ca.weak.agreement}, then this is evidence that our abstraction via axiomatic logic has worked its magic and let us turn \emph{complex informal verbal arguments} into \emph{routine, regular symbolic ones}.
This flavour of boredom is something that the field of distributed algorithms (in my opinion) needs more of.

And this simplicity via abstraction in logic is not mathematically trivial; far from it. 
The use of three truth-values, semitopologies, modal logic, and the precise forms of the axioms, are novel.
It was hardly obvious that this combination of tools could be usefully applied to study distributed algorithms.
Yet we have seen that this is so, and once we have set the machinery up, axioms become compact and precise, and proofs become --- if not completely easy --- at least precise, symbolic, compact, and clear.
\item
\emph{This is not a paper on distributed programming,} it is a paper on distributed \emph{algorithms}. 
\quad
Process algebras (perhaps guided by session types), or imperative programs with (say) concurrent separation logic, are valid approaches to designing distributed programs, but so far as I am aware they are no help to design (say) an efficient blockchain consensus algorithm. 
Those approaches are also typically operational in nature, aiming to model and reason about the dynamic behaviour of systems; whereas the approach proposed in this paper is more static, focussing on design-level evaluation rather than runtime behaviour.
While that is true, the fundamental difference of process calculi, session types, and similar systems is that they are fundamentally concerned with writing correct programs, whereas the focus in this paper (and in distributed algorithms in general) is on designing specific algorithms, like `agreement' or `broadcast', that are resilient to faulty behaviour.\footnote{A supplementary word here: \emph{algorithms} and \emph{programs} live on fundamentally different levels of the development stack.  Quicksort is a classic list sorting algorithm; \texttt{l.sort()} is a (line in) a program; these are not equivalent.  Thinking of algorithms as `just special cases of programs' is a category error that puts the cart before the horse; similarly \emph{distributed algorithms} are not just a special case of \emph{distributed programs}.}
This is a different part of the design space, with its own distinctive literature and norms.
\end{enumerate} 

\subsection{Reflection on the axiomatisations}

\begin{enumerate*}
\item
\emph{Developing an axiomatisation can be hard work, even for a `simple' protocol} (though it gets easier with practice).

The work is constructive in the sense that it comes about because we are forced to understand the protocol more deeply.
An example of this (which I did not expect) was the discovery of the accidental redundant conjunct in the Crusader Agreement protocol (Remark~\ref{rmrk.ca.discuss}(\ref{item.ca.error})).
This was effectively invisible in the pseudocode description, but became clear in the declarative logical axiomatisation.

It is easy enough to read a protocol and to fool ourselves into thinking that we understand it.\footnote{Though speaking for myself, I could not even manage that.  When I was first exposed to these distributed protocols --- written in some combination of English and pseudocode --- I did not understand them \emph{at all}.  %
Abstracting away the pseudocode to logical axioms was my way of making sense of the field.
Axioms do not lie: they mean what they mean. 
If you mean something else, then you tweak the syntax or change the logic.} 
Converting it into an axiomatisation forces us to be explicit about the nature of the thing itself.
\item
I would argue that \emph{the pseudocode in presentations like Figure~\ref{fig.ca.code} is not even really code, or even pseudocode.}

It is written in code-like style, but consider that clauses are `executed' in parallel and they follow a Horn clause like structure: Horn clauses are declarative and logical in flavour (coming from logic programming), not imperative and procedural. 
So I would argue that the logical content is already there, even in Figure~\ref{fig.ca.code}, and it is just struggling to escape from a code-flavoured presentation. 
\item
\emph{How do we know when an axiomatisation is `good'?}

One useful test is: does it make the proofs easy?
As is often the case, small (or large!) differences in how axioms are presented can make big differences to how cleanly proofs go through.
This is a criterion for preferring one axiomatisation over another, and it is also a way to understand the protocol: we can say that we have truly understood a distributed protocol when we have expressed it as an axiomatisation that is simple and gives simple proofs. 
Finding this axiomatisation is often not trivial, but once found, its simplicity is powerful evidence that we have arrived at something \emph{mathematically essential} about what the protocol is trying to express. 
\item 
\emph{Axiomatic proofs are very precise} (certainly compared to English or pseudocode descriptions).

A predicate has an unambiguous meaning, and logical proofs are just about manipulating axioms using the rules of logic.
This can be a helpful support for precision and correctness.
\item
\emph{How canonical is the modal logic in this paper?}
Might other logics be required for other protocols?

The modal logic in this paper is canonical in the sense that it seems to be minimal (this is an intuition, not a theorem) required to express Bracha Broadcast and Crusader Agreement.
This is nice, because it gives some precise formal measure of how Bracha Broadcast and Crusader Agreement are similar: same logic; same models; slightly different axioms and correctness properties.

However, the logic in this paper is certainly not sufficient for all possible protocols. 
The Declarative Paxos axiomatisation~\cite{gabbay:decasd} requires a different notion of possible world (a participant and a round number) \emph{and} a different logic, featuring distinctive `Paxos-flavoured' modalities.
Heterogeneous Bracha Broadcast requires its own highly distinctive set of models and modalities~\cite{gabbay:hettrb,hart:hetbl4}.

In practice, I expect each family of protocols to induce a notion of model and modal logic.
Family resemblances between protocols will be reflected in structural and syntactic similarities (or identities) between their models and logics.
I would argue that this is a feature, since it brings out mathematical structure that is not always evident from reading an English description. %
So in this sense, several distinct natural logics emerging from several distinct protocols or protocol-families is useful. 

In contrast, if we want to build tools (as we mention in Future Work below), or if we want to study classes of protocol-families, we might prefer to construct \emph{the} encompassing, all-singing-all-dancing parametric framework into which all (or at least many) protocol logics might fit, the better to study families of logics and the better to aid users in setting up \emph{their} favourite protocol in a convenient engineering environment.
This is a different challenge and we make a start in that direction with the \emph{Coalition Logic} framework in~\cite{gabbay:decasd}.
\end{enumerate*}

\subsection{Algorithmic time, logical time, and causation}
\label{subsect.algorithmic.time}

We have noted (e.g. in the Abstract, or just above in Subsection~\ref{subsect.related.work}) that the logical approach of this paper treats a distributed algorithm declaratively as an axiomatic theory: and therefore, it does \emph{not} treat the algorithm as a transition system.

Unlike with approaches based on LTL or TLA+ or similar logics, there is no transition system of an abstract machine; there are just possible worlds (which in this paper are just participants), and predicates taking truth-values on these worlds.
This invites a question: what notions of time and causation can exist in the declarative analysis, if any?

Here, the phrases \emph{algorithmic time} and \emph{logical time} may be helpful: our axiomatisations eliminate algorithmic time (there are no state transitions) but we shall see that there can still be a notion of logical time, and indeed in a certain sense logical time is the true notion of time, and algorithmic time is just a side-effect of implementation.

Let us examine the axiomatisations: the voting protocol in Figure~\ref{fig.simple}; Bracha Broadcast in Figure~\ref{fig.bb}; and Crusader Agreement in Figure~\ref{fig.ca}.
We note that the backward rules can be viewed as putting predicates in a logical time ordering.
For example in Figure~\ref{fig.bb}, we note that:
\begin{itemize*}
\item
\rulefont{BrEcho?} can be read as `if I echo, it is \emph{because} somebody broadcast', and 
\item
\rulefont{BrReady?} can be read as `if I am ready, it is \emph{because} a quorum echoed', and
\item
\rulefont{BrDeliver?} can be read as `if I deliver, it is \emph{because} a quorum was ready'.
\end{itemize*}
These are intuitions, not mathematical proofs: the axioms are what they are, but we can interpret them as above.
Similar patterns are observable in Figures~\ref{fig.ca} and Figure~\ref{fig.simple}.

In this view, the backward rules say that
broadcast causes/comes before echo, which causes/comes before ready, which causes/comes before deliver.
This is logical time and logical causation.%
\footnote{Do the forward rules now also give notions of logical time and causation?  
Yes they do --- they reflect the same structure --- but arguably not quite as clearly, just because forward rules tend to be more complicated.

Intuitively this is because moving \emph{forwards}, we have to account for failure or branches, whereas looking \emph{backward} we know that something has happened and we just need to explain why.  
This shows up even in the simple examples of this paper, e.g. in Figure~\ref{fig.bb} we have three backward rules but \emph{four} forward rules, and in Figure~\ref{fig.ca} we have the same number of backward as forward rules but the forward rules have slightly more complex structure (an extra disjunct in \rulefont{CaEcho1!}; an additional existential quantifier in \rulefont{CaEcho2!}).  
So if all we care to do is identify logical sequencing, then backward rules may give a cleaner view.
} 

A declarative axiomatisation of Paxos~\cite{gabbay:decasd} has a similar structure, except that it is richer, reflecting the additional complexity of the Paxos algorithm.
Declarative Paxos requires a \emph{round number} $n$, which is attached to each possible world.
So there, the notion of logical time is created by an interplay between the structure of axioms ordering predicates (as above), and the explicit round number carried by each possible world.\footnote{If the reader prefers, we could make a finer distinction: between \emph{logical time} and \emph{causation}.  In this view, we could say that Bracha Broadcast and Crusader Agreement have no notion of logical time but they do have a notion of causation as discussed; in contrast, Paxos does have a notion of logical time given by the round number, and within each round it also has a notion of causation which resembles the one we see in this paper.  
I find it simplest to just think in terms of a single notion of logical time, referring to the logical structure of how the algorithm evolves.}

I expect that this will be typical: each family of protocols will come with its own notion of logical time,
which will be expressed by a combination of explicit structure on the set of possible worlds, and implicit structure visible in the backward rules.

To sum up: the declarative approach eliminates algorithmic time (by design) but retains an abstract notion of \emph{logical time}.
I would further argue that logical time is in fact more faithful to what time \emph{essentially is} in the algorithm.
Any algorithm needs to be implemented, and thus it needs to be embedded in an implementation-dependent algorithmic time, but from the point of view of this paper, algorithmic time is a byproduct of implementation, and it is not something that is necessarily essential to what the algorithm \emph{is}.%
\footnote{This point is not purely philosophical.  Eliminating algorithmic time leads to smaller models (because we do elide implementational detail of what transitions happened in which order).  Therefore if and when model-checking tools might become available for the declarative analysis, we can anticipate that eliding algorithmic time may yield smaller search-spaces for counterexamples.}

This may be a different perspective on what time is from what the reader may be used to seeing, but I think it is an interesting one and natural in its own way, once we get used to the declarative approach.

\subsection{Future work}

\subsubsection{Other algorithms}
\quad
Natural future work is to study these ideas applied to other distributed algorithms.
Ones that are foundational for many modern blockchain systems include: Paxos (for this, a draft paper exists~\cite{gabbay:decasd}), Practical Byzantine Fault Tolerance (PBFT)~\cite{DBLP:conf/osdi/CastroL99}, HotStuff~\cite{DBLP:conf/podc/YinMRGA19}, and Tendermint~\cite{DBLP:journals/corr/abs-1807-04938}; 
another family of protocols organises transactions into a directed acyclic graph (DAG) rather than a strict chain (for parallelism and scalability), including DAG-rider~\cite{DBLP:conf/podc/KeidarKNS21}, Narwhal and Tusk~\cite{DBLP:conf/eurosys/DanezisKSS22}, and Bullshark~\cite{DBLP:conf/ccs/SpiegelmanGSK22}. 

Many other protocols and protocol-families exist, and studying what logics and families of logics they might correspond to, is an open problem whose exploration would be well worthwhile.\footnote{To pick one of many recent examples: at time of writing, Solana had recently proposed \emph{Alpenglow}~\url{https://forum.solana.com/t/simd-0326-proposal-for-the-new-alpenglow-consensus-protocol/4236}.} 

\subsubsection{Heterogeneous Paxos}

In other research, we have used declarative techniques to find errors in Heterogeneous Paxos~\cite{sheff:hetp,sheff_et_al:LIPIcs.OPODIS.2020.5}, and to design a replacement protocol~\cite{gabbay:hettrb} whose correctness has been verified in Lean~4~\cite{hart:hetbl4}.
Empirically, we found that the power of declarative techniques played a key role: \emph{the relevant design-space exceeded our ability to explore using pseudocode and intuition}, and our understanding of the protocol was greatly enhanced by having access to declarative abstractions.
A full account of this scaled-up application of declarative techniques so far is in~\cite{gabbay:hettrb}, and this is ongoing work.

\subsubsection{More than three truth-values}

We used three truth-values in this paper: $\tvF$ and $\tvT$ model correct behaviour, and $\tvB$ models faulty behaviour.
That sufficed for voting, Bracha Broadcast, and Crusader Agreement, but there are many kinds of faulty behaviour and in other applications we may want to distinguish them: e.g. we might wish to distinguish \emph{crash faults} (where a participant crashes permanently or for a while) from \emph{byzantine faults} (where a participant actively interacts with other participants not according to the protocol).

Could we use a four-valued logic for this?
What about a five-valued logic?
What about truth-values from the interval $[0,1]$?
What about an arbitrary lattice, or Heyting algebra?

The answer to all of these questions is: yes.
It is our logic, and we can do what we want.
In particular, we can use whichever domain of truth-values we feel most appropriately captures what we want to express --- be it two-valued, three-valued (as has been helpful in this paper), or more.\footnote{%
We do need to be careful about a possible explosion in the number of possible connectives.

Two-valued logic has a space of unary connectives of size $2^2=4$.
The space of binary connectives has size $2^{(2^2)}=16$.
This is manageable, but the number grows: $n$-valued logic has a space of unary connectives of size $n^n$, and a space of binary connectives of size $n^{(n^2)}$.
If $n=3$ (as is the case of this paper) then there are $3^{9}\approx 2*10^4$ binary connectives.
If $n=4$ then we have $4^{16}=2^{32}\approx 4*10^9$.
If $n=5$ then we have $5^{25}\approx 3*10^{17}$.

So we should not add truth-values just for the sake of it.
However, we should also not hesitate to use a particular domain of truth-values (even a large one) if we have \emph{specific} need of it. 
This is because \emph{if} we have a particular application in mind, which requires us to use a particular domain of truth-values, then that application should also pick out the relevant modalities and connectives for us, and we are free to ignore the (possibly large) space of other connectives. 
}

\subsubsection{Mechanised tools}

Other natural future work is to integrate the logics thus produced into model-checkers and theorem-provers.
There is nothing necessarily mathematically complex about this (which is a feature, not a bug): a three-valued logic over a semitopology is not a mathematically complex entity.
For example: the TLC model checker~\cite{DBLP:conf/charme/YuML99} for TLA+ provides state exploration to verify safety and liveness properties; the PlusCal translator simplifies specification by offering a high-level pseudocode-like interface; and TLA+ Toolbox~\cite{DBLP:journals/corr/abs-1912-10633} provides a development environment.
It would be good to integrate the logic from this paper into this toolchain or one like it.
If there is a catch here, it is that semitopologies involve powersets, which can lead to an exponential state space explosion when looking for counterexamples, so we might do well to choose a tool that would handle this well.
On this topic semitopologies have an algebraic analogue of \emph{semiframes}~\cite{gabbay:semasa,gabbay:semdca};
these provide additional abstraction and so might be useful in the search for counterexamples, by somewhat reducing state space.
 
\subsubsection{Implementation}
\label{subsect.implementation}

We have seen in this paper how distributed protocols can be represented as axiomatic theories.
Quorums are modelled via semitopologies; byzantine behaviour is modelled via three-valued logic; axioms are written in a simple but expressive modal logic.

By design, this elides implementational details, so the question arises: can we put these details back in? 
Can we go from abstract specification to concrete implementation?
And if we can, could we go to a \emph{fast} implementation?

I see these questions as being essentially about, or analogous to, questions of compilation.
That is: given any declarative program or specification, can we compile it to transitions of a concrete machine --- embedded in `algorithmic time', as per Subsection~\ref{subsect.algorithmic.time} --- and if so, can we make that machine run quickly?

I am extremely hopeful that this should be possible, for two reasons. 
First, we can already do similar things with programming. 
For instance, `efficient compilers for functional programming languages' is not an oxymoron; on the contrary it is a lucrative profession.
And second, when we look at the axioms of \ThyBB and \ThyCA we see that they are highly structured and almost like clauses in some kind of logic programming language or rewrite system.
This suggests that while \emph{for arbitrary theories} (i.e. random sets of axioms) execution may not be efficient or even possible, protocol designers are not generating protocols corresponding to arbitrary theories.
On the contrary, they are generating structured theories that may consist of predicates satisfying some kind of syntactic form not dissimilar to Horn clauses, or perhaps ones which can be seen as confluent rewriting theories.
There is every prospect that these may be obligingly susceptible to automated generation of efficient executables.
Finding out what this syntactic form may be, and constructing appropriate compilers, would be interesting future work.

\newcommand\urlprefix{}

\subsection*{Acknowledgements}

Thanks to two anonymous referees for their attention to this material.
Thanks to Luca Zanolini for long conversations on distributed systems.
I am particularly grateful to Jan Mas Rovira, for choosing to formalise this paper in Lean~\cite{rovira:hbbl4}, for noting various typos and errors in the original material, and for many questions and observations which led to improvements in the exposition.
I am also grateful to the organisers of and participants in Dagstuhl Seminar 26071 on \emph{Behavioural Types for Resilience}, for giving me the opportunity to present some ideas from this material at the seminar \emph{and} for detailed discussions and questions which again led to improvements in the exposition.

\end{document}